\documentclass[a4paper,11pt]{article}

\usepackage[utf8x]{inputenc}
\usepackage{amsmath,amsfonts,amssymb,amstext,amsthm}
\usepackage{graphicx}
\usepackage{a4wide}
\usepackage{enumerate}

\usepackage{tikz}
\usetikzlibrary{calc}
\usetikzlibrary{shapes.symbols}

\usepackage{algorithm}
\usepackage{filecontents}
\usepackage[noend]{algpseudocode}

\newcommand{\parp}{\ensuremath{\oplus\mathrm{P}}}

\newcommand{\Ptime}{\ensuremath{\mathrm{P}}}
\newcommand{\NPtime}{\ensuremath{\mathrm{NP}}}
\newcommand{\nump}{\ensuremath{\mathrm{\#P}}}
  
  \newcommand\pin{p}  
\newcommand{\prb}[1]{\textsc{#1}}
\newcommand{\pinnedparhcol}{\ensuremath{\oplus \prb{PinnedHomsTo}H}}

\newcommand{\parhcol}[1][H]{\ensuremath{\oplus\prb{HomsTo}{#1}}}
\newcommand{\paris}{\ensuremath{\oplus\prb{IS}}}
\newcommand{\partlabparhcol}{\ensuremath{\oplus \prb{PartLabHomsTo}H}}

\newcommand{\dist}{\mathrm{dist}}
\newcommand{\Aut}{\mathrm{Aut}}
\newcommand{\Orb}{\mathrm{Orb}}
\newcommand{\HomPin}{\mathrm{PinHom}}

\newcommand{\dom}{\mathrm{dom}}

\newcommand{\calG}{\ensuremath{\mathcal{G}}}
\newcommand{\calI}{\ensuremath{\mathcal{I}}}
\newcommand{\calX}{\ensuremath{\mathcal{X}}}

\newcommand{\NCl}{\ensuremath{\#C_\ell}}

\newcommand{\Oy}{\ensuremath{\Omega_y}}
\newcommand{\Oz}{\ensuremath{\Omega_z}}
\newcommand{\Sox}{\ensuremath{\Sigma_{o,x}}}
\newcommand{\Sos}{\ensuremath{\Sigma_{o,s}}}
\newcommand{\Six}{\ensuremath{\Sigma_{i,x}}}
\newcommand{\Sis}{\ensuremath{\Sigma_{i,s}}}

\newcommand{\vecv}{\mathbf{v}}
\newcommand{\xbar}{\bar{x}}
\newcommand{\ybar}{\bar{y}}

\newcommand{\isoto}{\cong}

\newcommand{\eqclass}[1]{[\hspace{-0.2em}[{#1}]\hspace{-0.2em}]}

\newcommand{\Homs}[2]{\mathrm{Hom}({#1}\to {#2})}
\newcommand{\InjHoms}[2]{\mathrm{InjHom}({#1}\to {#2})}

\newcommand{\Lovasz}{Lov\'asz}
\newcommand{\Nesetril}{Ne\v{s}et\v{r}il}

\newtheorem{theorem}{Theorem}[section]
 
 \newtheorem{lemma}[theorem]{Lemma}
 \newtheorem{corollary}[theorem]{Corollary}
 \newtheorem{conjecture}[theorem]{Conjecture}

 \theoremstyle{definition}
 \newtheorem{definition}[theorem]{Definition}

\title{Counting Homomorphisms to Square-Free Graphs, Modulo~2%
\thanks{
The research leading to these results has received funding from 
the European Research Council under the European Union's Seventh Framework Programme (FP7/2007--2013) ERC grant agreement no.\ 334828. The paper 
reflects only the authors' views and not the views of the ERC or the European Commission. The European Union is not liable for any use that may be made of the information contained therein.
 Authors' address: Department of Computer Science, University of Oxford, Wolfson Building, Parks Road, Oxford, OX1~3QD, UK.}}
\author{Andreas G\"obel \and Leslie Ann Goldberg \and David Richerby}
\date{\today}
\begin{document}

\maketitle

\begin{abstract}
    We study the problem \parhcol{} of counting, modulo~$2$, the
    homomorphisms from an input graph to a fixed undirected graph~$H$.
A characteristic feature of modular counting is that cancellations make wider classes of instances tractable than is the case for exact (non-modular) counting, so subtle dichotomy theorems can arise.    
       We show the following dichotomy: for any $H$ that contains no
    $4$-cycles, \parhcol{} is either in polynomial time or
    is \parp{}-complete.  This partially confirms a conjecture of Faben and
    Jerrum that was previously only known to hold for trees and 
for a restricted class of tree-width-$2$ graphs called    
cactus graphs.  
We confirm the conjecture for 
a rich class of graphs including graphs of unbounded tree-width. In particular, 
we focus on  square-free graphs, which are graphs without 
$4$-cycles.    These graphs arise frequently 
in combinatorics, for example in connection with the strong perfect graph theorem and in 
certain graph algorithms.
Previous dichotomy theorems required the graph to be tree-like so that
tree-like decompositions could be exploited in the proof. We prove the conjecture
for a much richer class of graphs by adopting a much more general approach.
\end{abstract}

\section{Introduction}

 A homomorphism from a graph~$G$ to a graph~$H$ 
is a function from~$V(G)$ to~$V(H)$ that preserves edges, in the sense of mapping  
every edge of~$G$ to an edge of~$H$; non-edges of~$G$ may be mapped to edges or non-edges of~$H$.
Many structures arising in graph theory
can be represented naturally as homomorphisms. 
For example, the proper $q$-colourings of a graph~$G$
correspond to the homomorphisms from~$G$ 
to a $q$-clique. 
For this reason, homomorphisms from $G$ to a graph~$H$ are
often called ``$H$-colourings'' of~$G$.
Independent sets of~$G$
correspond to the homomorphisms from~$G$ to
the connected graph with two vertices  and one self-loop
(vertices of~$G$ which are mapped to the self-loop are out
of the corresponding independent set; vertices which are
mapped to the other vertex are in it).
Homomorphism problems can also be
seen as constraint satisfaction problems (CSPs) in which the
constraint language consists of a single symmetric binary relation.
Partition functions in statistical physics 
such as the Ising model, the Potts model, and the
hard-core model arise naturally as weighted
sums of homomorphisms~\cite{BG, GGJT}.

In this paper, we study the
complexity of counting homomorphisms modulo~$2$. 
For graphs $G$ and $H$,   $\Homs{G}{H}$ denotes 
the set of homomorphisms from~$G$ to~$H$. 
For each fixed~$H$,
we study the computational problem~\parhcol,
which is the problem of computing~${\left|\Homs{G}{H}\right|} \bmod 2$,
given an input graph~$G$.
 
The structure of the graph~$H$
strongly influences the complexity of~$\parhcol$.  
For example, consider the graphs~$H_1$ and~$H_2$ 
in Figure~\ref{fig:example}.
Our result (Theorem~\ref{thm:main}) shows
that \parhcol[H_1] is \parp{}-complete,
whereas \parhcol[H_2] is in~\Ptime{}.

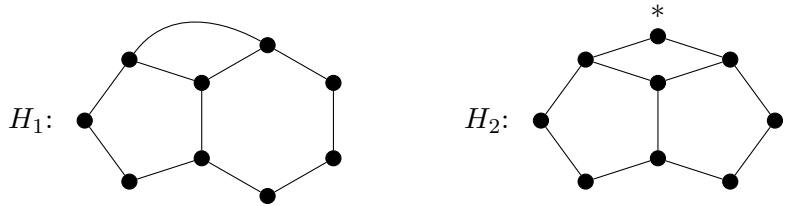
\begin{figure}
\begin{center}

\begin{tikzpicture}[scale=1,node distance = 1.5cm]
\tikzstyle{dot}   =[fill=black, draw=black, circle, inner sep=0.15mm]
\tikzstyle{vertex}=[fill=black, draw=black, circle, inner sep=2pt]
\tikzstyle{dist}  =[fill=white, draw=black, circle, inner sep=2pt]
\tikzstyle{pinned}=[draw=black, minimum size=11mm, circle]

    % H_1
    \begin{scope}[shift={(-3,0)}]
        \node at (-2.25,0) {$H_1$:};

        \node[vertex] (a) at (0,-0.5) {};
        \node[vertex] (b) at (0, 0.5) {};
        \node[vertex] (c) at ($(b)+(162:1)$) {};
        \node[vertex] (d) at ($(c)+(234:1)$) {};
        \node[vertex] (e) at ($(d)+(306:1)$) {};

        \node[vertex] (f) at ($(b)+( 30:1)$) {};
        \node[vertex] (g) at ($(f)+(330:1)$) {};
        \node[vertex] (h) at ($(g)+(270:1)$) {};
        \node[vertex] (i) at ($(h)+(210:1)$) {};

        \draw (a) -- (b) -- (c) -- (d) -- (e) --
              (a) -- (i) -- (h) -- (g) -- (f) -- (b);
        \draw (c) .. controls ($(c)+(54:1)$) and ($(f)+(150:0.4)$) .. (f);
    \end{scope}

    % H_2
    \begin{scope}[shift={(3,0)}]
        \node at (-2.25,0) {$H_2$:};

        \node[vertex] (a) at (0,-0.5) {};
        \node[vertex] (b) at (0, 0.5) {};
        \node[vertex] (c) at ($(b)+(162:1)$) {};
        \node[vertex] (d) at ($(c)+(234:1)$) {};
        \node[vertex] (e) at ($(d)+(306:1)$) {};

        \node[vertex] (f) at ($(b)+( 18:1)$) {};
        \node[vertex] (g) at ($(f)+(306:1)$) {};
        \node[vertex] (h) at ($(g)+(234:1)$) {};

        \node[vertex] (i) at ($(c)+( 18:1)$) [label=90:$*$] {};

        \draw (a) -- (b) -- (c) -- (d) -- (e) --
              (a) -- (h) -- (g) -- (f) -- (b);
        \draw (c) -- (i) -- (f);
    \end{scope}

\end{tikzpicture}
\end{center}
\caption{
Theorem~\ref{thm:main} shows
that \parhcol[H_1] is \parp{}-complete,
whereas \parhcol[H_2] is in~\Ptime{}.
This, and the role of the starred vertex are explained later in the introduction.}
\label{fig:example}
\end{figure}

 The aim of research in this area 
 is to understand for which graphs~$H$ the problem \parhcol{} is in~\Ptime,
 for which graphs~$H$  the problem  is \parp{}-complete, and to prove
 that, for all graphs~$H$, one or the other is true.
 Note that it isn't obvious, a priori, that there are no graphs~$H$
 for which \parhcol{} has intermediate complexity -- proving that
 there are no such graphs~$H$ is the main work of a so-called \emph{dichotomy theorem}.
 
This line of work was introduced by Faben and Jerrum~\cite{FJ13}.
They made the following important conjecture (which requires
a few definitions to state).
An \emph{involution} of a graph is an
automorphism of order~$2$, i.e., an  
automorphism~$\rho$ that is not the identity but for which $\rho^2$ is the identity.
Given a graph~$H$ and an involution~$\rho$,
$H^\rho$ denotes the subgraph of~$H$ induced by the fixed points of~$\rho$.
We write $H \Rightarrow H'$ if
there is an involution~$\rho$ of~$H$ such that $H^\rho=H'$
and we write $H \Rightarrow^* H'$
if either $H$ is isomorphic to~$H'$ (written
$H\isoto H'$)
or, for some positive integer~$k$,  there are graphs $H_1, \dots, H_k$ such that 
$H \isoto H_1$,
$H_1 \Rightarrow \cdots \Rightarrow H_k$, and
$H_k  \isoto H'$.
Faben and Jerrum showed \cite[Theorem 3.7]{FJ13}
that for every graph~$H$ there is (up to isomorphism) exactly one 
involution-free graph~$H^*$
such that  $H \Rightarrow^* H^*\!$. This graph~$H^*$ is called the 
\emph{involution-free reduction} of~$H$.
See \cite[Figure 1]{FJ13} for a diagram showing a graph  
 being reduced to its involution-free reduction. Faben and Jerrum make the following conjecture.

\begin{conjecture} \label{conj:FJ} (Faben and Jerrum~\cite{FJ13})
  Let $H$ be a  graph.  If its involution-free
    reduction~$H^*$ has at most one vertex, then \parhcol{} is in
    \Ptime{}; otherwise, \parhcol{} is \parp{}-complete.\end{conjecture}

Note that our claim in Figure~\ref{fig:example} is consistent with
Conjecture~\ref{conj:FJ}.
$H_1$ is involution-free, so it is its own involution-free reduction,
but the involution-free reduction of~$H_2$ is the single vertex marked~$*$ in the figure.

Faben and Jerrum \cite[Theorem 3.8]{FJ13} proved Conjecture~\ref{conj:FJ} for
the case in which~$H$ is a tree.
Subsequently, the present authors \cite[Theorem 1.6]{GGR14:Cactus}
proved the conjecture for
a well-studied class of tree-width-$2$ graphs, namely 
\emph{cactus graphs}, which are  
graphs in which each edge belongs to at most one cycle.

The main result of this paper is to prove the conjecture
for a much richer class of graphs.
In particular, we prove the conjecture for
every graph~$H$ whose involution-free reduction has no $4$-cycle
(whether induced or not).

Graphs without $4$-cycles are called ``square-free'' graphs.
These graphs arise frequently in combinatorics, for example in connection with the
strong perfect graph theorem~\cite{SF} and certain graph algorithms~\cite{SF2}.
Our main theorem is the following.

\newcommand{\statethmmain}{
    Let $H$ be a graph whose involution-free reduction~$H^*$ is
    square-free.  If $H^*$~has at most one vertex, then \parhcol{} is
    in \Ptime{}; otherwise, \parhcol{} is \parp{}-complete.}
\begin{theorem}
\label{thm:main}
\statethmmain{}
\end{theorem}

If $H$ is square-free, then so is every induced subgraph, including its involution-free
reduction~$H^*\!$. Thus, we have the following corollary.

\begin{corollary}
\label{cor:main} 
Let $H$ be a square-free graph.  If its involution-free
    reduction~$H^*$ has at most one vertex, then \parhcol{} is in
    \Ptime{}; otherwise, \parhcol{} is \parp{}-complete.
    \end{corollary}

In Section~\ref{sec:whysquares} we will 
discuss the reasons that we require $H^*$ to be square-free
in the proof of Theorem~\ref{thm:main}. First, 
in Section~\ref{sec:intro:parity},
we will describe the
background to counting modulo~$2$.
In Section~\ref{sec:beyond}, 
we will explain why Conjecture~\ref{conj:FJ} is so much more difficult to prove
for graphs with unbounded tree-width.
Very briefly, in order to prove  
that \parhcol{} is \parp{}-hard
without having a bound on the tree-width of~$H$,
it is necessary to take a much more abstract approach.
Since it is not possible to decompose~$H$
using a tree-like decomposition as we did in~\cite[Theorem 1.6]{GGR14:Cactus},
we have instead come up with an abstract
characterisation of
graph-theoretic structures in~$H$ 
which   lead to \parp{}-hardness.
As we shall see, the proof that such structures always exist in square-free graphs involves interesting non-constructive
elements, leading to a more abstract, and less technical (graph-theoretic)
proof than~\cite{GGR14:Cactus}, 
while applying to a substantially richer set of graphs~$H$, including graphs
with unbounded tree width.

\subsection{Counting modulo~2}
\label{sec:intro:parity}

Although counting modulo~$2$ produces a one-bit answer, the complexity
of such problems has a rather different flavour from 
the complexity of 
decision
problems.  The complexity class \parp{} was first studied by
Papadimitriou and Zachos~\cite{PZ82:Counting} and by Goldschlager and
Parberry~\cite{GP86:Parallel}.
$\parp{}$ consists of all problems of the form
``compute $f(x) \bmod 2$'' where computing $f(x)$ is a problem in \nump.
Toda~\cite{Tod91:PP-PH} has
shown that there 
is a randomised polynomial-time reduction
from every
problem in the polynomial hierarchy to 
some problem in~\parp.
As
such, \parp{}~is a large complexity class and \parp{}-completeness
seems to represent a high degree of intractability.

The unique flavour of modular counting 
is exhibited by Valiant's famous 
restricted version of $3$-SAT~\cite{Val06:Accidental} for
which counting solutions is \nump{}-complete~\cite{XZZ07:3-regular},
counting solutions modulo~$7$ is in polynomial-time but counting
solutions modulo~$2$ is \parp{}-complete~\cite{Val06:Accidental}.  The
seemingly mysterious number~$7$ was subsequently explained by Cai and
Lu~\cite{CL11:Holographic}, who showed that the $k$-SAT version of
Valiant's problem is tractable modulo any prime factor of $2^k-1$.

Counting modulo~$2$ closely resembles ordinary, non-modular
counting, but is still very different.  Clearly, if a counting problem can
be solved in polynomial time, the corresponding decision and parity
problems are also tractable, but the converse does not necessarily
hold.  A characteristic feature of modular counting is cancellations,
which can make the modular versions of hard counting problems
tractable.  For example, consider not-all-equal SAT, the
problem of assigning values to Boolean variables such that each of a
given set of clauses contains both true and false literals.  The
number of solutions is always even, since solutions can be paired up
by  negating every variable in one solution to obtain a
second solution.  This makes counting modulo~$2$ trivial, while
determining the exact number of solutions is
\nump{}-complete~\cite{GGL14:Locally-optimal} and even deciding 
whether
a
solution exists is \NPtime{}-complete~\cite{Sch78:Satisfiability}.

We use cancellations extensively in this paper.  For example, if we
wish to compute the size of a set~$S$ modulo~$2$
then, for any even-cardinality
subset~$X\subseteq S$, we have $|S|\equiv |S\setminus X|\bmod2$.  This
means that we can ignore the elements of~$X$.    It is
also helpful to
 partition the set~$S$    into
disjoint subsets $S_1, \dots, S_\ell$ 
exploiting the fact that $|S|$ is congruent modulo~$2$ to  the number of odd-cardinality~$S_i$.
We use this idea frequently.
 
For work  on counting modulo~$k$ in the \emph{constraint satisfaction} setting
see~\cite{ghlx}.

\subsection{Going beyond bounded tree-width} 
\label{sec:beyond}
 
 \subsubsection{Trees}
 
All known hardness results for counting homomorphisms modulo~$2$
start with the following basic ``pinning'' approach.
Let $\pin$ be a function  from $V(G)$
to $2^{V(H)}$.
A homomorphism $f\in \Homs{G}{H}$ \emph{respects}
the pinning function~$\pin$
if, for every $v\in V(G)$, 
$f(v)$ is in the set~$\pin(v)$.   
Let $\HomPin(G,H,\pin)$ be the set of homomorphisms
from~$G$ to~$H$ that respect the pinning function~$\pin$
and let \pinnedparhcol\ be the problem of counting, modulo~$2$,
the number of homomorphisms in $\HomPin(G,H,\pin)$,
given an input graph~$G$ and a pinning function~$\pin$.

Faben and Jerrum~\cite[Corollary~4.18]{FJ13} 
give a polynomial-time Turing 
reduction from the problem $\pinnedparhcol$ to the problem $\parhcol$
for the special case in which the pinning function
pins only two vertices of~$G$,
and these are both pinned to entire orbits of the automorphism
group of~$H$. 
The reduction relies on a result of \Lovasz{}~\cite{Lov67:OpStruct}.
 
In order to use the reduction, it is necessary to show that
the special case of the problem $\pinnedparhcol$
is itself \parp{}-hard.  
Faben and Jerrum restrict their attention to the case in which $H$ is a tree,
and this is helpful. Every involution-free tree is asymmetric (so the orbit
of every vertex is trivial),
so the pinning function~$\pin$ is actually able to pin two vertices of~$G$ 
to any two \emph{particular} vertices of~$H$.

 The reduction that they used to prove hardness 
 of $\pinnedparhcol$
 is
from \paris{}, the problem of counting independent sets modulo~$2$,
which was shown to be \parp{}-complete by
Valiant~\cite{Val06:Accidental}.

We first give an informal description of
a general reduction from   \paris{}
to the problem $\pinnedparhcol$.   (The general description is
actually based on our current approach in this paper, but
we can also present past approaches 
in this context.)
The vertices and edges of an input~$G$
of \paris{} are replaced by gadgets to give a graph~$J$.  In~$J$, the
gadget corresponding to the vertex~$v$ of~$G$ has a vertex~$y^v$.
We also choose an appropriate vertex~$i$ in~$H$.  Any
homomorphism~$\sigma$ from $J$ to the target graph~$H$ defines a set
$I(\sigma) = \{v\in V(G)\mid \sigma(y^v)=i\}$ (mnemonic: ``$i$'' means
``in'' because $\sigma(y^v)$ is~$i$ exactly when $v$ is in $I(\sigma)$).  
The configuration of the gadgets ensures that  
a set $I \subseteq V(G)$ has an odd number of homomorphisms
$\sigma$ with $I(\sigma)=I$ if
and only if $I$ is an independent set of~$G$.
Next, the homomoprhisms $\sigma\in\Homs{J}{H}$ can be partitioned 
  according to the value of~$I(\sigma)$.  By the
partitioning argument mentioned at the end of
Section~\ref{sec:intro:parity}, the number of independent sets in~$G$
is equivalent to $|\Homs{J}{H}|$, modulo~$2$.

The gadgets are chosen according to the structure and properties
of~$H$.  Since Faben and Jerrum 
were working
with trees, they  were able to use gadgets with
very simple structure: their gadgets are essentially paths and they
exploit the fact that any non-trivial involution-free tree has at
least two even-degree vertices 
and, of course, these 
have a unique path between them (which turns out to be useful).

 \subsubsection{Cactus graphs}

The situation for cactus graphs is much more complicated.  Non-trivial
involution-free cactus graphs still contain even-degree vertices but
the presence of cycles means that paths, even shortest paths, are no
longer guaranteed to be unique.  
Our solution in~\cite{GGR14:Cactus}
was to use more complicated gadgets.  They are still (loosely) based on paths,
since they are defined in terms of numbers of walks between vertices of~$H$.
However,  rather than
requiring appropriate even-degree vertices (which might not exist), 
we used a second, and more complicated, gadget to ``select'' an even-cardinality subset of a vertex's
neighbours.  To find such gadgets in~$H$, 
we used tree-like decompositions.
Given a decomposition that breaks $H$ into independent
fragments, we inductively found gadgets (or, sometimes, partial gadgets)
in the fragments, carefully putting them together across the join
of the decomposition. All of this led to a very technical, very graph-theoretic solution, 
and also to a solution that does not generalise to graphs  without tree-like decompositions.

The proof is  complicated by the fact that there are involution-free  graphs
(even involution-free cactus graphs!)\@
that have non-trivial automorphisms, unlike the situation for trees. 
Thus, the fact that the pinning function
pins vertices to entire orbits
(rather than to particular vertices) 
causes complications. The solution in~\cite[Section 8]{GGR14:Cactus}
relies on special properties of cactus graphs, and it is not clear how it could be
generalised.
  
\subsubsection{Unbounded tree-width}

Since they are based around a tree-like decomposition, the
techniques of~\cite{GGR14:Cactus}
are not suitable for graphs with unbounded tree-width.
To prove  Conjecture~\ref{conj:FJ} for a richer class of graphs,
we adopt a much more abstract approach. 
Since we do not have tree-like decompositions,
we instead mostly  use structural properties of
the whole graph to find gadgets.  
The structural properties do not always require technical detail --
as we will see below, re-examining a result of
\Lovasz{}~\cite{Lov67:OpStruct}  even allows us to demonstrate non-constructively
the existence of some of the gadgets that we use.
 
 In order to support our more general approach, we first have to
modify
the pinning problem \pinnedparhcol.
 For any graph~$H$, a \emph{partially $H$-labelled graph} $J=(G,\tau)$
consists of an \emph{underlying graph}~$G$ and a \emph{pinning
function}~$\tau$, which in this paper is a partial function from $V(G)$ to~$V(H)$.
Thus, every vertex $v$ in the domain of~$\tau$ is pinned to
a \emph{particular} vertex of~$H$ and \emph{not} to a subset such as an orbit.
 A homomorphism from a partially labelled graph~$J=(G,\tau)$ to~$H$ is
a homomorphism $\sigma\colon G\to H$ such that, for all vertices $v\in
\dom(\tau)$, $\sigma(v) = \tau(v)$.
The intermediate problem that we study then is \partlabparhcol{},
the problem of computing $|\Homs{J}{H}| \bmod 2$,
given a partially $H$-labelled graph~$J$.
In Section~\ref{sec:pinning}, we generalise the application of \Lovasz's theorem to
show (Theorem~\ref{thm:partlabcol}) that
$\partlabparhcol\leq \parhcol$.

 Armed with a stronger pinning 
 technique, we then abstract away most of the complications 
 that arose
 for graphs with small tree-width
 by instead
using more general gadgets, defined in Section~\ref{sec:gadgets}.
Because they are not based on paths, they do not rely on uniqueness of
any path in~$H$.  Instead, the gadgets have three main parts.  Our new
reduction from \paris{} to \parhcol{} can be seen informally as
assigning colours to both the vertices and the edges of~$G$, where each
``colour'' is a vertex of~$H$.  One part of the gadget controls which
colours can be assigned to  each vertex, one controls which colours
can be assigned to each edge and a third part determines how many
homomorphisms there are from $G$ to~$H$, given the choice of colours
for the vertices and edges.  
In addition to all of this, we identify two special 
vertices of~$H$, one of which is the
vertex~$i$ mentioned above.

The much more general nature of our gadgets compared to those used
previously makes them much easier to find and, in some cases, allows
us to prove the existence of parts of them non-constructively.\footnote{
Recall that gadgets depend only on the fixed graph~$H$ and not on the input~$G$
so they can be hard-coded into the reduction --- there is no need to find one
constructively.}  
We no longer need to
find unique shortest paths in~$H$ or, indeed, any paths at all. In fact,
all the gadgets that we construct in this paper use  a
``caterpillar gadget'' (Definition~\ref{defn:caterpillar}) which
allows us to use \emph{any} specified path in the graph~$H$ instead of 
relying on a unique shortest path.
Rather than finding hardness gadgets in components in some
decomposition of~$H$, we mostly find gadgets ``in situ''.  

When a graph has
two even-degree vertices, we can directly use those vertices and a
caterpillar gadget to produce a hardness gadget (see  Lemma~\ref{lem:two-even}).  This already
provides a self-contained proof of Faben and Jerrum's dichotomy for
trees.  Next, for graphs with only one even-degree vertex, we show 
(Corollary~\ref{cor:one-even-asym}) that
deleting an appropriate set of vertices leaves a component with two
even-degree vertices and show (Lemma~\ref{lem:even-deg})
how to simulate that vertex deletion
with gadgets.  This leaves only graphs in which every vertex has odd
degree.  In such a graph, we are able to use any shortest odd-length cycle to
construct a gadget (Lemma~\ref{lem:odd-cycle}). If there are no odd cycles, the graph is bipartite.
In this interesting case (Lemma~\ref{lem:always-even-gadget})
we  use our version of \Lovasz's result to find a gadget
non-constructively.

\subsection{Squares}
\label{sec:whysquares}

It is natural to ask
why the involution-free reduction~$H^*$ in 
Theorem~\ref{thm:main} is required to be square-free.
 We do not believe that the
restriction to square-free graphs is fundamental, since our results on
pinning apply to all involution-free graphs
(Section~\ref{sec:pinning}) and neither our definition of hardness
gadgets (Definition~\ref{defn:hardness-gadget}) nor our proof that the
existence of a hardness gadget for~$H$ implies that \parhcol{}
is \parp{}-complete (Theorem~\ref{thm:hardness-gadget}) requires $H$
to be square-free.  However, all the actual hardness gadgets that we find
for graphs do rely on the absence of $4$-cycles, as discussed in
Section~\ref{sec:gadgets:squares}, and removing this restriction seems
technically challenging.  We note that dealing with $4$-cycles also
caused significant difficulties in cactus graphs~\cite{GGR14:Cactus}.

\subsection{Related work}

We have already mentioned earlier work  on counting graph homomorphisms modulo~$2$.
The problem of counting graph homomorphisms
(exactly, rather than modulo a fixed constant) 
was previously studied by Dyer and
Greenhill~\cite{DG00:Homomorphisms}. They showed the problem of counting 
homomorphisms to a fixed graph~$H$ is solvable in polynomial time if 
every connected component of
$H$~is
a complete graph with a self-loop on every vertex or a complete  bipartite graph
with no self-loops, and is \nump{}-complete, otherwise.
Their work builds on an earlier dichotomy by Hell and \Nesetril{}~\cite{HN90:Hcol}
for the complexity of the graph homomorphism decision problem (the problem of distinguishing between the case
where there are no homomorphisms and the case where there is at least one).

\subsection{Organisation}

We introduce notation in Section~\ref{sec:notation}.
Section~\ref{sec:pinning} deals with pinning and consists mostly of
adapting existing work to the precise framework we require.  It can be
skipped by the reader who is comfortable with pinning and happy to
believe it can be done in our more general setting.

The gadgets that we use are formally defined in Section~\ref{sec:gadgets},
where we also show that \parhcol{} is \parp{}-complete if $H$~is an
involution-free graph that has one of these gadgets.
Section~\ref{sec:gadgets:caterpillar} introduces a gadget that we use
extensively, but which requires $H$ to be square-free, as discussed in
Section~\ref{sec:gadgets:squares}.  In
Section~\ref{sec:finding-gadgets}, we show how to find hardness
gadgets for all square-free graphs and, in Section~\ref{sec:mainthm},
we tie everything together to prove the dichotomy theorem.

 \section{Notation}
\label{sec:notation}

We write $[n]$ for the set $\{1, \dots, n\}$.
For a set~$S$ and an element~$x$, we often write $S-x$ for $S\setminus
\{x\}$.

\paragraph{Graphs.}  In this paper, graphs are undirected and have no
parallel edges and no loops.  The one exception to this is that we
briefly allow loops in the proof of Lemma~\ref{lem:Lovasz} (this is
clearly stated in the proof).  Paths and cycles do not repeat
vertices; walks may repeat both vertices and edges.
The length of  a path or cycle is the number of edges that it contains.
The \emph{odd-girth} of a graph is the length of
its shortest odd-length cycle.  $\Gamma_G(v)$ is the set
of neighbours of a vertex~$v$ in~$G$.

We write $G\isoto H$ to indicate that graphs $G$ and~$H$ are isomorphic.
$\Aut(H)$ denotes the automorphism group of a graph~$H$.  An
\emph{involution} is an automorphism of order~$2$ (i.e., an
automorphism~$\rho$ that is not the identity such that $\rho\circ\rho$
is the identity).
$\Homs{G}{H}$ denotes the set of homomorphisms from a graph~$G$ to a
graph~$H$.

\paragraph{Partially labelled graphs.}
For any graph~$H$, a \emph{partially $H$-labelled graph} $J=(G,\tau)$
consists of an \emph{underlying graph}~$G$ and a \emph{pinning
function}~$\tau$, which is a partial function from $V(G)$ to~$V(H)$.
A vertex~$v$ in the domain of the pinning function is said to be
\emph{pinned} or \emph{pinned to $\tau(v)$}.
We will refer to these graphs as \emph{partially labelled
  graphs} where the graph~$H$ is clear from the context.
We sometimes write $G(J)$ and~$\tau(J)$ for the underlying graph and
pinning function of a partially labelled graph, respectively.
We write
partial functions as sets of pairs, for example, writing $\tau =
\{a\mapsto s,b\mapsto t\}$ for the partial function~$\tau$ with
$\dom(\tau) = \{a,b\}$ such that $\tau(a)=s$ and $\tau(b)=t$.

A homomorphism from a partially labelled graph~$J=(G,\tau)$ to~$H$ is
a homomorphism $\sigma\colon G\to H$ such that, for all vertices $v\in
\dom(\tau)$, $\sigma(v) = \tau(v)$.  We say that such a homomorphism
\emph{respects}~$\tau$.  

\paragraph{Distinguished vertices.}  
It is often convenient to regard a graph as having some number of
distinguished vertices $x_1, \dots, x_r$ and we denote such a graph by
$(G, x_1, \dots, x_r)$.  Note that the distinguished vertices need not
be distinct.
We sometimes abbreviate the sequence 
$x_1,\ldots,x_r$ as $\xbar$
and we use $G[\xbar]$ to denote the subgraph of~$G$
induced by the set of vertices $\{x_1,\ldots,x_r\}$.
A homomorphism from a graph $(G, x_1, \dots,
x_r)$ to $(H, y_1, \dots, y_r)$ is a homomorphism~$\sigma$ from $G$
to~$H$ with the property that $\sigma(x_i)=y_i$ for each $i\in[r]$.
This is the same thing as a homomorphism from the partially $H$-labelled
graph $(G, \{x_1\mapsto y_1, \dots, x_r\mapsto y_r\})$ to~$H$.
Given a partially
labelled graph~$J=(G,\tau)$ and vertices $x_1, \dots, x_r\notin
\dom(\tau)$, a homomorphism from $(J, x_1, \dots, x_r)$ to $(H,
y_1, \dots, y_r)$ is formally identical to a homomorphism from
$J'=(G, \tau \cup \{x_1\mapsto y_1, \dots, x_r\mapsto y_r\})$ to~$H$.

Similarly, we say that two graphs $(G, x_1, \dots, x_r)$ and $(H, y_1,
\dots, y_s)$ are isomorphic if $r=s$ and there is an isomorphism
$\rho\colon V(G) \to V(H)$ such that $\rho(x_i)=y_i$ for each
$i\in[r]$ (note that we may have $G=H$).  An automorphism of $(G, x_1,
\dots, x_r)$ is just an automorphism~$\rho$ of~$G$ with the property
that $\rho(x_i)=x_i$ for each $i\in[r]$.

\paragraph{Diagram conventions.} In diagrams of partially labelled graphs,
ordinary vertices are denoted by black dots, distinguished vertices by
small white circles and pinned vertices (i.e., the vertices in
$\dom(\tau)$) by large white circles.  A
label next to a vertex of any kind indicates the identity of that
vertex; a label inside a white circle indicates what that vertex is
pinned to.

\section{Partially labelled graphs and pinning}
\label{sec:pinning}

The results in this section do not require~$H$ to be square-free.

Because we use pinning in our gadgets, we mostly work with the problem of
determining the number of homomorphisms from a partially
$H$-labelled graph to~$H$, modulo~$2$:

\begin{description}
\setlength{\itemsep}{-0.9ex}
\item \emph{Name:} \partlabparhcol{}.
\item \emph{Parameter:} A graph~$H$.
\item \emph{Input:} A partially $H$-labelled graph $J$.
\item \emph{Output:} $|\Homs{J}{H}| \bmod 2$.
\end{description} 

Our goal in the remainder of this section is to prove the following
theorem.

\begin{theorem}
\label{thm:partlabcol}
    For any involution-free graph~$H$, $\partlabparhcol\leq \parhcol$.
\end{theorem}

The reader who is prepared to take Theorem~\ref{thm:partlabcol} on
trust may safely skip the rest of this section.  The theorem itself is
used in later sections but the details of its proof are not.

To prove the theorem, we need to develop some machinery.  This closely
follows the presentation of similar material by Faben and
Jerrum~\cite{FJ13} and our earlier paper~\cite{GGR14:Cactus} which, in
turn, draw on the work of \Lovasz~\cite{Lov67:OpStruct} and Hell and
\Nesetril~\cite{HN04:HomBook}.  This duplication is unfortunate but,
at the end of the section, we explain how the results we have
presented are subtly different from those in the literature
so existing results could not be reused directly.

After stating some elementary group theory results that we need, we
prove in Section~\ref{sec:partlab:Lovasz} a version of a result
originally due to \Lovasz{}. This (Lemma~\ref{lem:Lovasz}) states
that, if graphs with
distinguished 
vertices
$(H,\ybar)$ and~$(H'\!,\ybar')$ are 
non-isomorphic,
there is a graph $(G,\xbar)$ that has an odd number of homomorphisms
to one of $(H,\ybar)$ and~$(H'\!,\ybar')$ and an even number of
homomorphisms to the other.  Taking $H'=H$, this allows us to
distinguish two tuples of vertices in~$H$ from one another, as long as
they are not in the same orbit of $\Aut(H)$.

This is not quite enough for pinning, as it doesn't give us control
over which of the two graphs receives an odd number of homomorphisms
from $(G,\xbar)$.  In Section~\ref{sec:partlab:impvec}, we solve this
problem algebraically, adapting a technique of Faben and
Jerrum~\cite{FJ13}.  This allows us to prove
Theorem~\ref{thm:partlabcol} in Section~\ref{sec:partlab:pinning} and
thereby implement the pinning we need for our reductions.

\subsection{Group-theoretic background}
\label{sec:partlab:groups}

We will require two results from group theory.  For the first, see, e.g., 
\cite[Theorem~13.1]{Arm88:GroupSym}.

\begin{theorem}[Cauchy's group theorem]
\label{thm:Cauchy}
    If $\calG$ is a finite group and a prime~$p$ divides~$|\calG|$,
    then $\calG$~contains an element of order~$p$.
\end{theorem}

For a permutation group~$\calG$ acting on a set~$X$, the \emph{orbit}
of an element $x\in X$ is the set $\Orb_{\calG}(x) = \{\pi(x)\mid
\pi\in\calG\}$.  For a graph~$H$, we will abuse notation mildly by writing
$\Orb_H(\cdot)$ instead of $\Orb_{\Aut{H}}(\cdot)$.

The following is a corollary of the orbit--stabiliser
theorem \cite[Corollary~17.3]{Arm88:GroupSym}.

\begin{theorem}
\label{thm:OrbStab}
    Let $\calG$ be a finite permutation group acting on a set~$X$.
    For every $x\in X$, $\left|\Orb_{\calG}(x)\right|$ divides~$|\calG|$.
\end{theorem}

These two theorems have the following corollary about the size of
orbits under the automorphism group of involution-free graphs.

\begin{corollary}
\label{cor:odd-orbit}
    Let $H$~be an involution-free graph.  Every orbit of a tuple
    $\ybar\in V(H)^r$ under the action of $\Aut(H)$ has odd
    cardinality.
\end{corollary}
\begin{proof}
    By Theorem~\ref{thm:Cauchy}, $|\Aut(H)|$
    is odd, since the group
    contains no element of order~$2$.  Consider the natural action of
    $\Aut(H)$ on $V(H)^r\!$.  By Theorem~\ref{thm:OrbStab}, the size
    of the orbit of~$\ybar$ in~$H$ divides $|\Aut(H)|$ so is also odd.
\end{proof}

\subsection{A \Lovasz{}-style lemma}
\label{sec:partlab:Lovasz}

\Lovasz{} proved that two graphs $H$ and~$H'$ are isomorphic if and
only if $|\Homs{G}{H}| = |\Homs{G}{H'}|$ for every graph~$G$ (in
fact, he proved the analogous result for general relational structures
but we do not need this here).  We show that this result remains true
even if we replace equality of the number of homomorphisms with
equivalence modulo~$2$.  Faben and Jerrum also showed this 
\cite[Lemma~3.13]{FJ13}, though in a less general setting than the one
  that we need.
Our proof is based on the presentation of \cite[Section~2.3]{HN04:HomBook}.

For the proof we need some definitions, which are used only in this
section.  We say that two $r$-tuples $\xbar$ and~$\ybar$ \emph{have
  the same equality type} if, for all $i,j\in [r]$, $x_i=x_j$ if and
only if $y_i=y_j$.  Let $\InjHoms{(G,\xbar)}{(H,\ybar)}$ be the set of
injective homomorphisms from $(G, \xbar)$ to~$(H,\ybar)$.

Before proving the main lemma, we prove a simple fact about injective
homomorphisms and equality types of distinguished variables.

\begin{lemma}
\label{lem:eq-type}
    Let $(G,\xbar)$ and~$(H,\ybar)$ be graphs, each with
    $r$~distinguished vertices.  If $\xbar$ and~$\ybar$ do not have
    the same equality type, then $|\InjHoms{(G,\xbar)}{(H,\ybar)}| =
    0$.
\end{lemma}
\begin{proof}
    If there are $i,j\in[r]$ such that $x_i=x_j$ but $y_i\neq y_j$,
    then there are no homomorphisms (injective or otherwise) from
    $(G,\xbar)$ to $(H,\ybar)$, since $x_i$~cannot be mapped
    simultaneously to both $y_i$ and~$y_j$.  Otherwise, there must be
    $i,j\in[r]$ such that $x_i\neq x_j$ but $y_i=y_j$.  Then no
    homomorphism~$\eta$ can be injective because we must have
    $\eta(x_i) = \eta(x_j) = y_i$.
\end{proof}

\begin{lemma}
\label{lem:Lovasz}
    Let $(H, \ybar)$ and $(H'\!, \ybar')$ be involution-free graphs,
    each with $r$~distinguished vertices.  
Then
    $(H, \ybar) \isoto (H'\!,
    \ybar')$ if and only if, for all (not necessarily connected)
    graphs $(G,\xbar)$ with $r$~distinguished vertices,
    \begin{equation}
    \label{eq:Lovasz}
        |\Homs{(G,\xbar)}{(H,\ybar)}|
            \equiv |\Homs{(G,\xbar)}{(H'\!,\ybar')}| \pmod2\,.
    \end{equation}
\end{lemma}
\begin{proof}
    If $(H,\ybar)$ and $(H'\!,\ybar')$ are isomorphic, it follows
    trivially that \eqref{eq:Lovasz}~holds for all graphs $(G,\xbar)$.
    For the other direction, suppose that \eqref{eq:Lovasz}~holds for
    all $(G,\xbar)$.

    First, we claim that this implies that $\ybar$ and~$\ybar'$ have
    the same equality type.  If they have different equality types
    then, without loss of generality, we may assume that there are
    distinct indices $i$ and~$j$ such that $y_i=y_j$ but $y'_i\neq
    y'_j$.  Let $G$~be
    the graph on vertices $\{y_1, \dots, y_r\}$ with no edges: we see
    that $|\Homs{(G, \ybar)}{(H,\ybar)}| = 1 \neq
    |\Homs{(G,\ybar)}{(H'\!,\ybar')}|=0$, contradicting the assumption
    that \eqref{eq:Lovasz} holds for all~$G$.

    Second, we show by induction on the number of vertices in~$G$
    that, if \eqref{eq:Lovasz}~holds for all $(G,\xbar)$ then, for all
    $(G,\xbar)$,
    \begin{equation}
    \label{eq:injective}
        |\InjHoms{(G,\xbar)}{(H,\ybar)}|
            \equiv |\InjHoms{(G,\xbar)}{(H'\!,\ybar')}| \pmod2\,,
    \end{equation}
    Specifically, under the assumption that
    \eqref{eq:Lovasz}~holds for all~$(G,\xbar)$, we show that
    \eqref{eq:injective} holds for all $(G,\xbar)$ with $|V(G)|\leq
    n_0$ for a suitable value~$n_0$ and that, if
    \eqref{eq:injective} holds for all $(G,\xbar)$ with $|V(G)|<n$, it
    also holds for any $(G,\xbar)$ with $|V(G)|=n$.

    Let $n_0 = |\{y_1, \dots, y_r\}| = |\{y'_1, \dots, y'_r\}|$ be the
    number of distinct elements in $\ybar$.  For the base case of the
    induction, consider any graph $(G, \xbar)$ with $|V(G)|\leq
    n_0$.  If $\xbar$~does not have the same equality type as $\ybar$
    and~$\ybar'$ (which is guaranteed if $|V(G)| < n_0$) then, by
    Lemma~\ref{lem:eq-type},
    \begin{equation*}
        |\InjHoms{(G,\xbar)}{(H,\ybar)}|
            = |\InjHoms{(G,\xbar)}{(H'\!,\ybar')}| = 0\,.
    \end{equation*}
    If $\xbar$~has the same equality type as $\ybar$ and~$\ybar'\!$
    then, in particular, every vertex of~$G$ is distinguished.  Any
    homomorphism from $(G,\xbar)$ to $(H,\ybar)$ or~$(H'\!,\ybar')$ is
    injective so we have
    \begin{align*}
        |\InjHoms{(G,\xbar)}{(H,\ybar)}|
            &= |\Homs{(G,\xbar)}{(H,\ybar)}|           \\
            &= |\Homs{(G,\xbar)}{(H'\!,\ybar')}|       \\
            &= |\InjHoms{(G,\xbar)}{(H'\!,\ybar')}|\,,
    \end{align*}
    where the second equality is by the assumption that
    \eqref{eq:Lovasz}~holds for~$(G,\xbar)$.

    For the inductive step, let $n>n_0$ and assume that
    \eqref{eq:injective}~holds for all $(G,\xbar)$ with $|V(G)|<n$.
    Now, consider some $(G,\xbar)$ with $|V(G)|=n$.

    Given any homomorphism~$\sigma$ from $(G,\xbar)$ to $(H,\ybar)$, we
    can define an equivalence relation~$\theta$ on $V(G)$
    by $(u,v)\in \theta$ if and only if $\sigma(u)=\sigma(v)$.  (Note
    that, if $\sigma$~is injective, then $\theta$~is just the equality
    relation on~$V(G)$.)
    Write $\eqclass{u}$~for the $\theta$-equivalence class
    of a vertex $u\in V(G)$.  Let $G/\theta$ be the graph whose vertex
    set is $\{\eqclass{u}\mid u\in V(G)\}$ and whose edge set is
    $\{(\eqclass{u},\eqclass{v})\mid (u,v)\in E(G)\}$.  For graphs
    with distinguished vertices, we write $(G,x_1, \dots, x_r)/\theta =
    (G/\theta, \eqclass{x_1}, \dots, \eqclass{x_r})$. The
    homomorphism~$\sigma$ from $(G,\xbar)$ to $(H,\ybar)$ corresponds
    to an injective homomorphism from $(G,\xbar)/\theta$ to
    $(H,\ybar)$.

    Note that, if there are adjacent vertices $u$ and~$v$ in~$G$ such
    that $(u,v)\in\theta$ for some equivalence relation~$\theta$, the
    graph $G/\theta$ has a self-loop on the vertex~$\eqclass{u}$.
    This is not a problem.  Because $H$~is loop-free, there are no
    homomorphisms (injective or otherwise) from such a graph
    $G/\theta$ to~$H$.  For the same reason, there are no
    homomorphisms from $G$ to~$H$ that map adjacent vertices $u$
    and~$v$ to the same place.  Therefore, this particular~$\theta$
    does not correspond to any homomorphism from $G$ to~$H$ and
    contributes zero to the sums below, as required.

    We have
    \begin{align*}
        |\Homs{(G,\xbar)}{(H,\ybar)}|
            &= |\InjHoms{(G,\xbar)}{(H,\ybar)}|
                  + \sum_\theta |\InjHoms{(G,\xbar)/\theta}{(H,\ybar)}|\\
        |\Homs{(G,\xbar)}{(H'\!,\ybar')}|
            &= |\InjHoms{(G,\xbar)}{(H'\!,\ybar')}|
                  + \sum_\theta |\InjHoms{(G,\xbar)/\theta}{(H'\!,\ybar')}|\,,
    \end{align*}
    where the sums are over all equivalence
    relations~$\theta$, except for the equality relation.

    The left-hand sides of these equations are equivalent modulo~$2$
    by assumption.  The sums over~$\theta$ on the right are equivalent
    modulo~$2$ by the inductive hypothesis since $\theta$~is not
    the equality relation, so $G/\theta$~has fewer vertices than~$G$.
    Therefore, \eqref{eq:injective}~holds for the graph under
    consideration.

    Finally, it remains to prove that \eqref{eq:injective}~holding for
    all $(G, \xbar)$ implies that $(H,\ybar) \isoto (H'\!,\ybar')$.
    To see this, take $(G,\xbar) = (H,\ybar)$.  An injective
    homomorphism from a graph to itself is an automorphism and, 
since
 $(H,\ybar)$ is   involution-free,
$\Aut(H,\ybar)$
has no element of order~$2$, so
$|\Aut(H,\ybar)|$ is odd by Cauchy's
    group theorem (Theorem~\ref{thm:Cauchy}).
    By~\eqref{eq:injective}, there are an odd number
    of injective homomorphisms from $(H,\ybar)$ to $(H'\!,\ybar')$,
    which means that there is at least one such homomorphism.
    Similarly, taking $(G,\xbar) = (H'\!,\ybar')$ shows that there is
    an injective homomorphism from $(H'\!,\ybar')$ to $(H,\ybar)$ and,
    therefore, the two graphs are isomorphic.
\end{proof}

For our nonconstructive proof that some gadgets exist, we use the
following corollary of the proof of Lemma~\ref{lem:Lovasz}, which
restricts   to a certain class of connected graphs.  

\begin{corollary}\label{cor:Lovasz}
    Let $(H, \ybar)$ and $(H'\!, \ybar')$ be connected, involution-free 
     graphs, each with $r$~distinguished vertices, such that $H[\ybar]$
    and~$H'[\ybar']$ are also connected.  Then
    $(H, \ybar) \isoto (H'\!, \ybar')$ if and only if
    \eqref{eq:Lovasz}~holds for all connected graphs $(G,\xbar)$ with
    $r$~distinguished vertices such that $G[\xbar]$ is connected.
\end{corollary}
\begin{proof}
    For brevity, we refer to $(G,\xbar)$ as \emph{appropriate} if it
    is connected, it has $r$~distinguished vertices and $G[\xbar]$ is
    connected.  
    
    As 
in the proof of Lemma~\ref{lem:Lovasz},
the ``only if'' direction is trivial, 
 so we
 suppose that \eqref{eq:Lovasz} holds for all appropriate
    $(G,\xbar)$.
Also,    
    $\ybar$ and~$\ybar'$ must have the same equality
    type.  If they do not, we may assume there are distinct $i$
    and~$j$ with $y_i=y_j$ but $y'_i\neq y'_j$, and take
    $G = H[\ybar]$.  $(G,\ybar)$ is appropriate but we have
    $|\Homs{(G,\ybar)}{(H,\ybar)}| = 1 \neq
    |\Homs{(G,\ybar)}{(H'\!,\ybar')}| = 0$,
    which contradicts the assumption that \eqref{eq:Lovasz}~holds for
    all appropriate~$(G,\xbar)$.

    The proof that \eqref{eq:Lovasz}~holding for every appropriate~$G$
    implies that \eqref{eq:injective}~holds for every appropriate~$G$
    proceeds by induction on $|V(G)|$, as 
in the proof of the lemma.
The base cases are
    unchanged.  To see that the inductive step remains valid, let
    $(G,\xbar)$ be appropriate and let $\theta$ be any equivalence
    relation on~$V(G)$.  We claim that $(G,\xbar)/\theta$ is also
    appropriate.  
By construction, $(G,\xbar)/\theta$ has $r$ distinguished vertices.
    It is connected because it is the result of
    identifying vertices in a connected graph;
    $(G/\theta)[\eqclass{x_1}, \dots, \eqclass{x_r}]$ is connected for
    the same reason.

    This establishes that \eqref{eq:injective}~holds for all
    appropriate~$(G,\xbar)$.  Since $(H,\ybar)$ and $(H'\!,\ybar')$
    are both appropriate, we can complete the proof in the same way as
in the proof of Lemma~\ref{lem:Lovasz},
substituting each of these graphs in turn for~$(G,\xbar)$
    in~\eqref{eq:injective}.
\end{proof}

\subsection{Implementing vectors}
\label{sec:partlab:impvec}

The presentation in this section follows very closely that of Faben
and Jerrum~\cite{FJ13}, extended to $r$-tuples of distinguished
vertices.

\begin{definition}
\label{defn:lambdas}
Let $H$ be  an involution-free graph.
We refer to a list  $\ybar_1, \dots, \ybar_\lambda$ of 
elements of~$V(H)^r$ 
as an \emph{enumeration of~$V(H)^r$  up to isomorphism}
if,  for every $\ybar\in V(H)^r$, there is exactly one $i\in[\lambda]$ such that
$(H,\ybar)\isoto (H,\ybar_i)$.  \end{definition}

Note that the number~$\lambda$ of tuples in the enumeration depends
on~$H$.

\begin{definition}
    Let $(G,\xbar)$~be a graph with $r$~distinguished vertices.  We
    define the vector $\vecv_H(G, \xbar)\in\{0,1\}^\lambda$
    where, for each $i\in[\lambda]$, the $i$th component of
    $\vecv_H(G, \xbar)$ is given by
    \begin{equation*}
        \big(\vecv_H(G, \xbar)\big)_i
            \equiv |\Homs{(G,\xbar)}{(H,\ybar_i)}| \pmod2\,.
    \end{equation*}
    We say that $(G,\xbar)$ \emph{implements} this vector.
\end{definition}

Define $\oplus$ and~$\otimes$ to be, respectively, component-wise
addition and multiplication, modulo~$2$, of vectors in
$\{0,1\}^\lambda\!$.

\begin{lemma}
\label{lem:otimesvec}
    Let $\xbar = x_1 \dots x_r$ and let $(G_1,\xbar)$ and
    $(G_2,\xbar)$ be graphs such that $V(G_1)\cap V(G_2) = \{x_1,
    \dots, x_r\}$.  Then,
    \begin{equation*}
        \vecv_H(G_1\cup G_2, \xbar)
            = \vecv_H(G_1, \xbar) \otimes \vecv_H(G_2, \xbar)\,.
    \end{equation*}
\end{lemma}
\begin{proof}
    A function $\sigma\colon V(G_1)\cup V(G_2)\to V(H)$ is a
    homomorphism from $(G_1\cup G_2, \xbar)$ to~$(H,\ybar)$ if and
    only if, for each $i\in\{1,2\}$, the restriction of~$\sigma$ to
    $V(G_i)$ is a homomorphism from $(G_i,\xbar)$ to~$(H,\ybar)$.
\end{proof}

In contrast, given $(G_1,\xbar_1)$ and $(G_2, \xbar_2)$, it is not
obvious that there is a graph $(G,\xbar)$ such that $\vecv_H(G,\xbar)
= \vecv_H(G_1,\xbar_1) \oplus \vecv_H(G_2,\xbar_2)$.  Following Faben
and Jerrum~\cite{FJ13}, we side-step this issue by introducing a
formal sum of graphs.  Given graphs with distinguished vertices $(G_1,
\xbar_1), \dots, (G_t, \xbar_t)$, we define
\begin{equation*}
    \vecv_H\big((G_1,\xbar_1) + \dots + (G_t,\xbar_t)\big)
        = \vecv_H(G_1,\xbar_1) \oplus \dots \oplus \vecv_H(G_t,\xbar_t)
\end{equation*}
and we say that a vector $\vecv\in\{0,1\}^\lambda$ is
\emph{$H$-implementable} if it can be expressed as such a sum.

We require the following, which is essentially
\cite[Lemma~4.16]{FJ13}.

\begin{lemma}
    Let $S\subseteq \{0,1\}^\lambda$ be closed under $\oplus$
    and~$\otimes$.  If $1^\lambda\in S$ and, for every distinct
    $i,j\in[\lambda]$, there is a tuple $s=s_1\dots s_\lambda\in S$
    with $s_i\neq s_j$, then $S = \{0,1\}^\lambda\!$.
\end{lemma}

\begin{corollary}
\label{cor:implementable}
    Let $H$~be an involution-free graph.  Every
    $\vecv\in\{0,1\}^\lambda$ is $H$-implementable.
\end{corollary}
\begin{proof}
    Let $S$ be the set of $H$-implementable vectors.  $S$~is clearly
    closed under~$\oplus$, and is closed under~$\otimes$ by
    Lemma~\ref{lem:otimesvec}.  Let $G$~be the graph on vertices
    $\{x_1, \dots, x_r\}$, with no edges.  $1^\lambda$~is implemented
    by $(G,x_1, \dots, x_r)$, which has exactly one homomorphism to
    every $(H,\ybar_i)$.  Finally, for every distinct pair $i,j\in
    [\lambda]$, $(H,\ybar_i)$ and $(H,\ybar_j)$ are not isomorphic, by
    definition of the enumeration of 
$r$-tuples (up to isomorphism).  
Therefore, by
    Lemma~\ref{lem:Lovasz}, there is a graph $(G,\xbar)$ such that
    \begin{equation*}
        |\Homs{(G,\xbar)}{(H,\ybar_i)}|
            \not\equiv |\Homs{(G,\xbar)}{(H,\ybar_j)}| \pmod 2\,.
    \end{equation*}
    $(G,\xbar)$ implements a vector~$\vecv$ whose $i$th and $j$th
    components are different.
\end{proof}

\subsection{Pinning}
\label{sec:partlab:pinning}

We now have almost everything we need to prove
Theorem~\ref{thm:partlabcol}.  Recall the definition of an enumeration
$\ybar_1, \dots, \ybar_\lambda$ of $V(H)^r$ up to isomorphism
(Definition~\ref{defn:lambdas}).

\begin{lemma}
\label{lem:sumvec}
    Let $H$~be an involution-free graph and let $\ybar_1, \dots,
    \ybar_\lambda$ be an enumeration of $V(H)^r$ up to isomorphism.
    For any graph~$(G, \xbar)$ with $r$~distinguished vertices,
    \begin{equation*}
        |\Homs{G}{H}|
            \equiv \sum_{i\in[\lambda]} (\vecv_H(G,\xbar))_i \pmod2\,.
    \end{equation*}
\end{lemma}
\begin{proof}
    We have (for details see below),
    \begin{align*}
        \sum_{i\in[\lambda]} (\vecv_H(G,\xbar))_i\;
            &\equiv \sum_{i\in[\lambda]} |\Homs{(G,\xbar)}{(H,\ybar_i)}| \pmod2\\
            &\equiv \sum_{i\in[\lambda]}
                        |\Orb_H(\ybar_i)|\,
                            |\Homs{(G,\xbar)}{(H,\ybar_i)}| \pmod2 \\
            &= \sum_{i\in[\lambda]}\ 
                        \sum_{\ybar\in\Orb_H(\ybar_i)}
                            |\Homs{(G,\xbar)}{(H,\ybar)}| \\
            &= |\Homs{G}{H}|\,.
    \end{align*}
    The second equivalence modulo~$2$ is because all orbits have odd
    cardinality by Corollary~\ref{cor:odd-orbit} and multiplying the terms
    of the sum by odd numbers doesn't change the total, modulo~$2$.
    The first equality is because, for any $\ybar\in\Orb_H(\ybar_i)$,
    $|\Homs{(G,\xbar)}{(H,\ybar)}| = |\Homs{(G,\xbar)}{(H,\ybar_i)}|$.
    This is because composing a homomorphism from $(G,\xbar)$
    to~$(H,\ybar)$ with an isomorphism from $(H,\ybar)$
    to~$(H,\ybar_i)$ gives a homomorphism from $(G,\xbar)$
    to~$(H,\ybar_i)$.
    The final equality is because every homomorphism
    from $G$ to~$H$ must map~$\xbar$ to some tuple~$\ybar$
and 
(exactly) all such tuples are included exactly
    once in the double sum.
\end{proof}

We can now prove Theorem~\ref{thm:partlabcol}: for any involution-free
graph~$H$, \partlabparhcol{} is polynomial-time Turing-reducible
to \parhcol{}.

\begin{proof}[Proof of Theorem~\ref{thm:partlabcol}]
    Let $J=(G,\tau)$ be an instance of \partlabparhcol{}.  Let
    $\xbar=x_1, \dots, x_r$ be an enumeration of $\dom(\tau)$ and let
    $\ybar = y_1, \dots, y_r = \tau(x_i), \dots, \tau(x_r)$.  Moving
    from the world of partially $H$-labelled graphs to the equivalent
    view of graphs with distinguished vertices, we wish to compute
    $|\Homs{(G,\xbar)}{(H,\ybar)}|$, modulo~$2$.

    By definition of the 
enumeration (up to isomorphism) 
 $\ybar_1, \dots, \ybar_\lambda$,
     there is some~$p$ such that $(H,\ybar)\isoto (H,\ybar_p)$.  Let
    $\vecv$ be the vector that has a $1$~in position~$p$ and has
    $0$~in every other position.  By
    Corollary~\ref{cor:implementable}, $\vecv$~is implemented by some
    sequence $(\Theta_1, \xbar_1), \dots, (\Theta_t, \xbar_t)$ of
    graphs with $r$-tuples of distinguished vertices.

    For each $i\in[t]$, let $(G_i,\xbar)$ be the graph that results from
    taking the union of disjoint copies of $G$ and~$\Theta_i$ and
    identifying the $j$th element of~$\xbar$ with the $j$th element of
    $\xbar_i$ for each $j\in [t]$.  We have
    \begin{align*}
        \vecv_H(G, \xbar)\otimes \vecv
            &= \vecv_H(G,\xbar)\otimes
                   \vecv_H\big((\Theta_1,\xbar_1) + \dots
                                   + (\Theta_t,\xbar_t)\big) \\
            &= \bigoplus_{i\in[t]} \big(
                   \vecv_H(G,\xbar)\otimes \vecv_H(\Theta_i, \xbar_i)
               \big)\, \\
            &= \bigoplus_{i\in[t]} \vecv_H(G_i,\xbar)\,.
    \end{align*}
    Now, sum the components of the vectors on the two sides of the
    equation.  On the right, by Lemma~\ref{lem:sumvec}, we have a
    value congruent modulo~$2$ to
    $\sum_{i\in[t]} |\Homs{G_i}{H}|$.  This can be computed by making
    $t$~calls to an oracle for \parhcol{}, and $t$~is bounded above by
    a constant, since $H$~is fixed. On the left, we have,
    $|\Homs{(G,\xbar)}{(H,\ybar)}|$, modulo~$2$, which is what we wish
    to compute.
\end{proof}

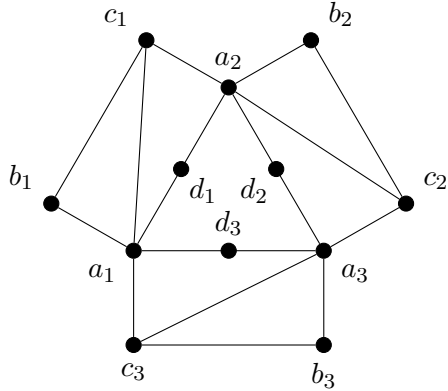
\begin{figure}
    \begin{center}
    \begin{tikzpicture}[scale=1.25]
        \tikzstyle{vertex}=[fill=black, draw=black, circle, inner
        sep=2pt]

        \node[vertex] (a1) at (0,0)              [label=210:$a_1$] {};
        \node[vertex] (a2) at (60:2)             [label= 90:$a_2$] {};
        \node[vertex] (a3) at (2,0)              [label=330:$a_3$] {};

        \node[vertex] (b1) at ($(a1) + (150:1)$) [label=150:$b_1$] {};
        \node[vertex] (b2) at ($(a2) + ( 30:1)$) [label= 30:$b_2$] {};
        \node[vertex] (b3) at ($(a3) + (270:1)$) [label=270:$b_3$] {};

        \node[vertex] (c1) at ($(a2) + (150:1)$) [label=150:$c_1$] {};
        \node[vertex] (c2) at ($(a3) + ( 30:1)$) [label= 30:$c_2$] {};
        \node[vertex] (c3) at ($(a1) + (270:1)$) [label=270:$c_3$] {};

        \node[vertex] (d1) at ($(a1) + ( 60:1)$) [label={[label distance=-4]-30:$d_1$}] {};
        \node[vertex] (d2) at ($(a2) + (300:1)$) [label={[label distance=-4]210:$d_2$}] {};
        \node[vertex] (d3) at ($(a3) + (180:1)$) [label={[label distance=-1]  90:$d_3$}] {};

        \draw (a1) -- (b1) -- (c1) -- (a1) -- (a2) -- (c1);
        \draw (a2) -- (b2) -- (c2) -- (a2) -- (a3) -- (c2);
        \draw (a3) -- (b3) -- (c3) -- (a3) -- (a1) -- (c3);
    \end{tikzpicture}
    \end{center}

    \caption{An involution-free graph~$H$ illustrating the difference between pinning
      vertices to orbits of vertices and pinning a tuple of vertices
      to an orbit of a tuple.}
    \label{fig:difference}
\end{figure}

The result we have proved appears similar to
\cite[Theorem~3.2]{GGR14:Cactus} but there is an important difference.
In~\cite{GGR14:Cactus}, we wished to pin $r$~vertices of~$G$, each to the orbit of a
vertex of~$H$.
In this paper, we focus on the problem  \partlabparhcol,
where we pin vertices of~$G$ to individual vertices of~$H$.
In order to achieve this, we essentially  pin an $r$-tuple of vertices of~$G$ to
the orbit of an $r$-tuple of vertices in~$H$.  
To see the difference,
consider the graph~$H$
in Figure~\ref{fig:difference}.
The orbits of single vertices are $\{a_1, a_2, a_3\}, \dots, \{d_1,d_2,d_3\}$.
There are six homomorphisms from the
single edge~$(x,y)$ to~$H$ that map~$x$ to the orbit of~$a_1$ and
$y$~to the orbit of~$d_1$ but only three that map the pair $(x,y)$ to
the orbit of the pair $(a_1,d_1)$, which is $\{(a_1,d_1), (a_2,d_2),
(a_3,d_3)\}$.

\section{Hardness gadgets}
\label{sec:gadgets}

In this section, we define gadgets that we will use to prove
\parp{}-completeness of \parhcol{} problems, by reduction from the
parity independent set problem $\paris$, i.e., the problem of
computing the number of independent sets in an input graph,
modulo~$2$.  $\paris$~was shown to be \parp{}-complete by
Valiant~\cite{Val06:Accidental}.

The gadgets we use are considerably more general than the ones we
defined for cactus graphs in~\cite{GGR14:Cactus}.  This allows us to
quickly prove hardness for large classes of square-free graphs and
even to find gadgets non-constructively.

In fact, our definition of hardness gadgets and the proof
that \parhcol{} is \parp{}-complete if $H$~is involution-free and has
a hardness gadget 
(Section~\ref{sec:gadgets:parp-complete})
does not
require the graphs to be square-free.  However, whenever we find a
gadget for a particular graph, it involves the ``caterpillar gadgets''
we introduce in Section~\ref{sec:gadgets:caterpillar}.  These gadgets
do depend on $H$~being square-free, as we show in
Section~\ref{sec:gadgets:squares}.

\subsection{$\parp$-completeness}
\label{sec:gadgets:parp-complete}

We now define the gadgets we use to prove hardness and show that they
serve this purpose.
Recall that a partially $H$-labelled graph $J$ consists of an underlying
graph $G(J)$ and a pinning function $\tau(J)$.
In the discussion that follows, we will choose a set $\Oy\subseteq
V(H)$ and a vertex $i\in\Oy$.  Given a graph~$G$ whose independent
sets we wish to count modulo~$2$, we will construct a partially
$H$-labelled graph~$J$ and consider homomorphisms from $J$ to~$H$.
$G(J)$~will contain a copy of $V(G)$ and we will be interested in
homomorphisms that map every vertex in this copy to~$\Oy$.  Vertices
mapped to~$i$ will be in the independent set under consideration;
vertices mapped to $\Oy-i$ will not be in the independent set.

\begin{definition}
\label{defn:hardness-gadget}
A \emph{hardness gadget} $(i,s,(J_1,y),(J_2,z),(J_3,y,z))$ for a
graph~$H$ consists
of vertices $i$ and~$s$ of~$H$ together with three connected, partially $H$-labelled
graphs with distinguished vertices 
$(J_1,y)$, $(J_2,z)$ and $(J_3,y,z)$ 
that satisfy certain
properties
as explained below.
Let
\begin{align*}
\Oy &= \{ a \in V(H) \mid |\Homs{(J_1, y)}{(H,a)}|\text{ is odd}\},\\
\Oz &= \{ b \in V(H) \mid |\Homs{(J_2, z)}{(H,b)}|\text{ is odd}\}, \mbox{ and}\\
\Sigma_{a,b} &=\Homs{(J_3,y,z)}{(H,a,b)}\,.
\end{align*}
The properties 
that
we require are 
the following.
\begin{enumerate}
\item $|\Oy|$ is even and $i\in \Oy$.
\item $|\Oz|$ is even and $s \in \Oz$.
\item 
For each $o\in \Oy-i$ and each $x\in \Oz-s$,
$|\Sox|$ is even.
\item
$|\Sis|$ is odd and, for each $o\in \Oy-i$ and each $x\in \Oz-s$,
$|\Sos|$ and $|\Six|$   are odd.
\end{enumerate}
\end{definition}

Before proving that hardness gadgets give \parp{}-completeness, we
introduce some notation.
Given partially $H$-labelled graphs $J_1=(G_1,\tau_1)$ and~$J_2=(G_2,\tau_2)$, with
$\dom(\tau_1)\cap \dom(\tau_2) = \emptyset$, we write $J_1\cup
J_2$ for the partially labelled graph~$J'=(G'\!,\tau')$, where $G' = G_1\cup G_2$
and $\tau' = \tau_1\cup \tau_2$.  That is,
$\dom(\tau') = \dom(\tau_1) \cup \dom(\tau_2)$ and
\begin{equation*}
    \tau'(v) = \begin{cases}
               \ \tau_1(v) &\text{ if } v\in\dom(\tau_1) \\
               \ \tau_2(v) &\text{ if } v\in\dom(\tau_2).
               \end{cases}
\end{equation*}

We will use the following notation to build partially labelled graphs
containing many copies of some subgraph.  For any ``tag''~$T$ (which
we will treat just as an arbitrary string) and any partially labelled
graph~$J$, denote by $J^T$ a copy of~$J$ with every vertex $v\in
V(G(J))$ renamed~$v^T\!$.

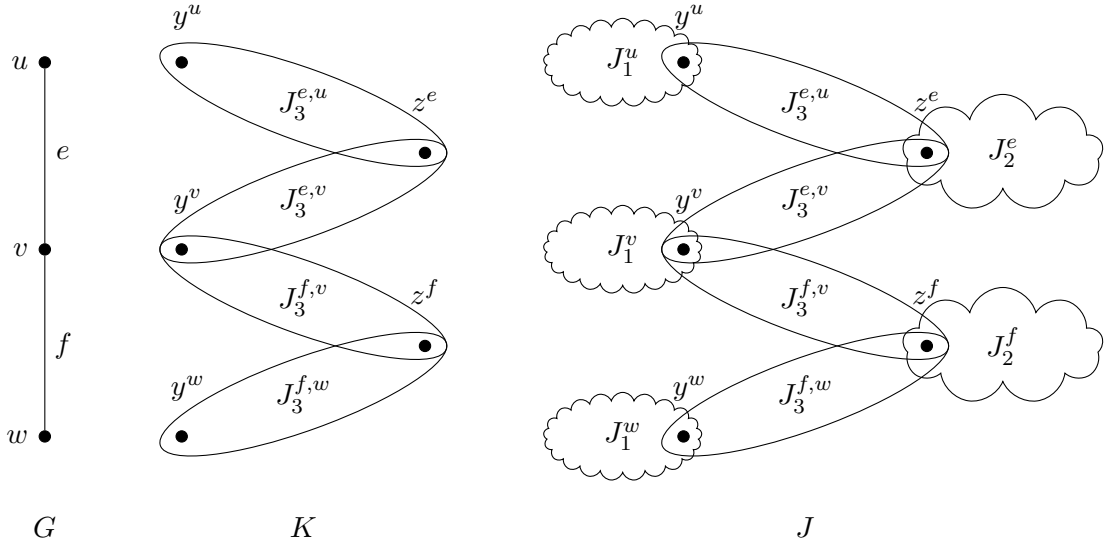
\begin{figure}
\begin{center}
\begin{tikzpicture}[scale=.4,node distance = 1.5cm]
\tikzstyle{vertex}=[fill=black, draw=black, circle, inner sep=1.5pt]
\tikzstyle{blob1}=[cloud, cloud puffs=21, draw, minimum width=.1cm,
                   minimum height=.1cm, aspect=2.1]
\tikzstyle{blob2}=[cloud, cloud puffs=9, draw, minimum width=1cm,
                   minimum height=.5cm, aspect=2.7]

    \begin{scope}[shift={(-25,0)}]
        \node[vertex] (u) at (0,   1.4) [label=180:{$u$}] {};
        \node             at (0.6,-1.6) {$e$};
        \node[vertex] (v) at (0,  -4.8) [label=180:{$v$}] {};
        \node             at (0.6,-8  ) {$f$};
        \node[vertex] (w) at (0,  -11 ) [label=180:{$w$}] {};
 
        \draw (u)--(v)--(w);
 
        \node at (0,-14) {$G$};
    \end{scope}

    \begin{scope}[shift={(-16.5,0)}]
        \draw (0, 0  ) circle [x radius=5cm, y radius=1.2cm, rotate=-20];
        \draw (0,-3.2) circle [x radius=5cm, y radius=1.2cm, rotate=20];
        \draw (0,-6.4) circle [x radius=5cm, y radius=1.2cm, rotate=-20];
        \draw (0,-9.6) circle [x radius=5cm, y radius=1.2cm, rotate=20];

        \node at (0, 0  ) {$J_3^{e,u}$};
        \node at (0,-3.2) {$J_3^{e,v}$};
        \node at (0,-6.4) {$J_3^{f,v}$};
        \node at (0,-9.6) {$J_3^{f,w}$};

        \node[vertex] at (-4  , 1.4) {};
        \node         at (-3.8, 2.9) {$y^u$};
        \node[vertex] at (-4  ,-4.8) {};
        \node         at (-3.8,-3.2) {$y^v$};
        \node[vertex] at (-4  ,-11 ) {};
        \node         at (-3.8,-9.4) {$y^w$};

        \node[vertex] at (4,-1.6) {};
        \node         at (4, 0.1) {$z^e$};
        \node[vertex] at (4,-8  ) {};
        \node         at (4,-6.2) {$z^f$};
        
                \node         at (0,-14 ) {$K$};
    \end{scope}

    \draw (0, 0  ) circle [x radius=5cm, y radius=1.2cm, rotate=-20];
    \draw (0,-3.2) circle [x radius=5cm, y radius=1.2cm, rotate=20];
    \draw (0,-6.4) circle [x radius=5cm, y radius=1.2cm, rotate=-20];
    \draw (0,-9.6) circle [x radius=5cm, y radius=1.2cm, rotate=20];

    \node at (0, 0  ) {$J_3^{e,u}$};
    \node at (0,-3.2) {$J_3^{e,v}$};
    \node at (0,-6.4) {$J_3^{f,v}$};
    \node at (0,-9.6) {$J_3^{f,w}$};

    \node[vertex] at (-4  , 1.4) {};
    \node         at (-3.8, 2.9) {$y^u$};
    \node[vertex] at (-4  ,-4.8) {};
    \node         at (-3.8,-3.2) {$y^v$};
    \node[vertex] at (-4  ,-11 ) {};
    \node         at (-3.8,-9.4) {$y^w$};
 
    \node[blob1] at (-6, 1.4) {$\quad\quad\quad$};
    \node        at (-6, 1.4) {$J_1^u$};
    \node[blob1] at (-6,-4.8) {$\quad\quad\quad$};
    \node        at (-6,-4.8) {$J_1^v$};
    \node[blob1] at (-6,-11 ) {$\quad\quad\quad$};
    \node        at (-6,-11 ) {$J_1^w$};
 
    \node[vertex] at (4,-1.6) {};
    \node         at (4, 0.1) {$z^e$};
    \node[vertex] at (4,-8  ) {};
    \node         at (4,-6.2) {$z^f$};
 
    \node[blob2] at (6.5,-1.6) {$\quad\quad\quad\quad$};
    \node        at (6.5,-1.6) {$J_2^e$};
    \node[blob2] at (6.5,-8  ) {$\quad\quad\quad\quad$};
    \node        at (6.5,-8  ) {$J_2^f$};
 
    \node at (0,-14) {$J$};
\end{tikzpicture}
\end{center}
\caption{The construction of the partially labelled graphs $K$ and~$J$
  from an example graph~$G$, as in the proof of
  Theorem~\ref{thm:hardness-gadget}.}
\label{fig:hardness-gadget}
\end{figure}

\begin{theorem}
\label{thm:hardness-gadget}
    If an involution-free graph $H$ has a hardness gadget then
    \parhcol{} is $\parp$-complete.
\end{theorem}
\begin{proof}
    Let $(i,s,(J_1,y),(J_2,z),(J_3,y,z))$ be the hardness gadget
    for~$H$ and recall the sets $\Oy$ and $\Oz$ from
    Definition~\ref{defn:hardness-gadget}.  We show how to reduce
    $\paris$ to \partlabparhcol{}; the result then follows from
    Theorem~\ref{thm:partlabcol} and \parp{}-completeness of
    $\paris$~\cite{Val06:Accidental}.  Given an input graph~$G$ to
    $\paris$, we construct an appropriate partially $H$-labelled
    graph~$J$ and show that $|\calI(G)| \equiv |\Homs{J}{H}| \bmod2$,
    where $\calI(G)$ is the set of independent sets in~$G$.

    We construct~$J$ in two stages (see
    Figure~\ref{fig:hardness-gadget}).  
    Take the union of disjoint
    copies~$J_3^{e,v}$ of~$J_3$ for every edge $e\in G$ and each
    endpoint~$v$ of~$e$. For each edge $e=(u,v)\in G$, identify the
    vertices $z^{e,u}$ and~$z^{e,v}$  
and call this~$z^e$.
    For each vertex $v\in G$,
    identify all the vertices~$y^{e,v}$ such that $e$~has $v$~as an
    endpoint,
and call this~$y^v$.
Call 
the resulting
graph~$K$.

    To make~$J$, take~$K$ and add a disjoint copy~$J_1^v$ of~$J_1$ for
    every vertex $v\in G$ and a disjoint copy~$J_2^e$ of~$J_2$ for
    every edge $e\in G$.  For each vertex $v\in G$, identify 
the vertex
 $y^v$ in~$K$ with the vertex $y^v$ in $J_1^v$.
 For each edge
    $e=(u,v)$ in~$G$, identify 
the vertex $z^e$ in~$K$ with the vertex $z^e$ in $J_2^e$.

    We now proceed to show that $|\Homs{J}{H}|\equiv |\calI(G)| \bmod2$.

    For a homomorphism $\sigma\in\Homs{K}{H}$, let $\eqclass{\sigma}$
    be the set of extensions of $\sigma$ to homomorphisms from $J$
    to~$H$, i.e.,
    \begin{equation*}
        \eqclass{\sigma} =
            \{\sigma'\in\Homs{J}{H}\mid
                  \sigma(v) = \sigma'(v) \text{ for all } v\in V(G(K))\}\,.
    \end{equation*}

    Every homomorphism from $J$ to~$H$ is the extension of a unique
    homomorphism from $K$ to~$H$, so we have
    \begin{equation}
    \label{eq:sum-eq-classes}
        |\Homs{J}{H}| \ \ \ = \!\!\!\!\!\!
            \sum_{\sigma\in\Homs{K}{H}} \!\!\!\!\!\! |\eqclass{\sigma}|\,.
    \end{equation}

    From the structure of~$J$, we have
    \begin{equation*}
        |\eqclass{\sigma}|
            = \left(\prod_{v\in V(G)} \big|\Homs{(J_1,y)}{(H,\sigma(y^v)}\big|\right)
              \left(\prod_{e\in E(G)} \big|\Homs{(J_2,z)}{(H,\sigma(z^e)}\big|\right)\,.
    \end{equation*}

    By Definition~\ref{defn:hardness-gadget},
    $|\Homs{(J_1,y)}{(H,a)}|$ is odd if and only if $a\in\Oy$ and
    $|\Homs{(J_2,z)}{(H,b)}|$ is odd if and only if $b\in\Oz$.
    Therefore, $|\eqclass{\sigma}|$~is odd if and only if
    $\sigma$~maps every vertex~$y^v$ into~$\Oy$ and every $z^e$
    into~$\Oz$: call such a homomorphism ``legitimate'' (with respect
    to $J_1$ and~$J_2$).  We can rewrite~\eqref{eq:sum-eq-classes} as
    \begin{equation}
    \label{eq:legitimate}
        |\Homs{J}{H}| \equiv |\{\sigma\in\Homs{K}{H}\mid
                               \sigma\text{ is legitimate}\} \pmod2\,,
    \end{equation}
    and, from this point, we restrict our attention to legitimate
    homomorphisms.

    Given a legitimate homomorphism $\sigma\in\Homs{K}{H}$, let
    $\sigma|_Y$~be the restriction of~$\sigma$ to the domain
    $\{y^v\mid v\in V(G)\}$.  Write $\sigma \sim_Y \sigma'$ if
    $\sigma|_Y = \sigma'|_Y$ and write $\eqclass{\sigma}_Y$ for the
    $\sim_Y$-equivalence class of~$\sigma$.  The
    classes~$\eqclass{\sigma}_Y$ partition the legitimate
    homomorphisms from $K$ to~$H$.  We have
    \begin{equation*}
        \big|\eqclass{\sigma}_Y\big| \ 
            = \!\!\!\!\!\!\prod_{(u,v)\in E(G)}\!\!\!\!\!\! n(\sigma(u), \sigma(v))\,,
    \end{equation*}
    where
    \begin{equation*}
        n(a,a') = \sum_{b\in \Oz}
                      \big|\Homs{(J_3,y,z)}{(H,a,b)}\big|\,
                          \big|\Homs{(J_3,y,z)}{(H,a'\!,b)}\big|\,.
    \end{equation*}
    By Definition~\ref{defn:hardness-gadget}, $|\Oz|$~is even, so the
    sum defining $n(a,a')$ has an even number of terms.
    $|\Homs{(J_3,y,z)}{(H,a,b)}|= |\Sigma_{a,b}|$ 
    is even if $a\in\Oy-i$ and
    $b\in\Oz-s$ and odd, otherwise.  If $a=a'=i$, every term is odd
    and $n(a,a')$ is even; otherwise, exactly one term ($b=s$) is odd,
    so $n(a,a')$ is odd.  Therefore, $|\eqclass{\sigma}_Y|$ is odd if
    and only 
  if  
    $\sigma$~does not map a pair of adjacent vertices to~$i$:
    that is, if the set $I(\sigma) = \{v\in V(G) \mid \sigma(y^v) =
    i\}$ is an independent set in~$G$.

    Choose representatives $\sigma_1, \dots, \sigma_k$, one from each
    $\sim_Y$-equivalence class.  We have
    \begin{align*}
           |\Homs{J}{H}|
            &\equiv |\{\sigma\in\Homs{K}{H}\mid
                           \sigma \text{ is legitimate}\}|\pmod2 \\
            &= \ \sum_{j=1}^k \big|\eqclass{\sigma_j}_Y\big|\\
            &\rule[-1.5ex]{0pt}{4ex}
               \equiv \ \big|\{j\in[k] \mid
                        I(\sigma_j)\text{ is independent}\}\big| \pmod2 \\
            &= \!\!\sum_{X\in \calI(G)}\!\!\!
                   \big|\{\sigma_j\mid j\in[k] \text{ and } I(\sigma_j)=X\}\big| \\
            &\equiv \ |\calI(G)| \pmod2\,,
    \end{align*}
    where the final equivalence is because the number of $\sigma_j$
    such that $I(\sigma)=X$ is exactly $|\Oy-i|^{|V(G)\setminus
      X|}$, which is odd because $|\Oy|$~is even.
\end{proof}

\subsection{Caterpillar gadgets}
\label{sec:gadgets:caterpillar}

All our hardness gadgets use the following ``caterpillar
gadgets'' as~$J_3$.  We will also use two further kinds of gadget,
``neighbourhood gadgets'' and ``$\ell$-cycle gadgets'', but we defer
their definitions to the sections where they are used.  As we will see
in the following section, caterpillar gadgets rely on $H$~being
square-free.

\begin{definition}
\label{defn:caterpillar}
    (See Figure~\ref{fig:caterpillar}).
    For a path $P=v_0\dots v_k$ in~$H$ of length at least~$1$,
    define the \emph{caterpillar gadget} $J_P=(G,\tau)$ as follows.
    $V(G)=\{u_1, \dots, u_{k-1}, w_1,\dots, w_{k-1},y,z\}$
    and $G$~is the path $yu_1\dots u_{k-1}z$ together with
    edges $(u_j,w_j)$ for $1\leq j\leq k-1$.
    $\tau=\{w_1\mapsto v_1, \dots, w_{k-1}\mapsto v_{k-1}\}$.
\end{definition}

\begin{figure}
\begin{center}
\begin{tikzpicture}[scale=1.5,node distance = 1.5cm]
\tikzstyle{dot}   =[fill=black, draw=black, circle, inner sep=0.15mm]
\tikzstyle{vertex}=[fill=black, draw=black, circle, inner sep=2pt]
\tikzstyle{dist}  =[fill=white, draw=black, circle, inner sep=2pt]
\tikzstyle{pinned}=[draw=black, minimum size=11mm, circle]

    % Main sequence of vertices
    \node[dist]   (y)  at (0  ,1) [label=90:$y$] {};
    \node[vertex] (u1) at (1  ,1) [label=90:$u_1$] {};
    \node[vertex] (u2) at (2  ,1) [label=90:$u_2$] {};
    \node[vertex] (u3) at (3.5,1) [label=90:$u_{k-2}$] {};
    \node[vertex] (u4) at (4.5,1) [label=90:$u_{k-1}$] {};
    \node[dist]   (z)  at (5.5,1) [label=90:$z$] {};

    % Pinned vertices
    \node[pinned] (w1) at (1  ,0) [label=-90:$w_1$] {$v_1$};
    \node[pinned] (w2) at (2  ,0) [label=-90:$w_2$] {$v_2$};
    \node[pinned] (w3) at (3.5,0) [label=-90:$w_{k-2}$] {$v_{k-2}$};
    \node[pinned] (w4) at (4.5,0) [label=-90:$w_{k-1}$] {$v_{k-1}$};

    % Solid edges
    \draw (y) -- (u2);
    \draw (z) -- (u3);

    \foreach \x in {1,2,3,4}
        \draw (u\x) -- (w\x);

    % Dotted edge
    \node[dot] (b) at ($(u2)!0.5!(u3)$) {};
    \node[dot] (a) at ($(b)!1mm!(u2)$) {};
    \node[dot] (c) at ($(b)!1mm!(u3)$) {};
    \draw (u2) -- ($(u2)!0.75!(a)$);
    \draw (u3) -- ($(u3)!0.75!(c)$);

\end{tikzpicture}
\end{center}
\caption{The caterpillar gadget corresponding to a path $v_0\dots
  v_k$.  The vertices $w_1, \dots, w_{k-1}$ in the gadget are pinned
  to vertices $v_1, \dots, v_{k-1}$ in~$H$, respectively.}
\label{fig:caterpillar}
\end{figure}
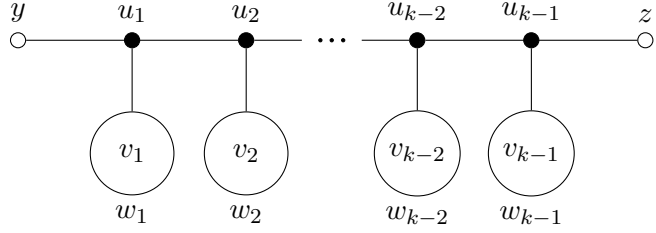

Note that, if $P$ is a single edge, $G(J_P)$~is also the single edge
$(y,z)$ and $\tau(J_P) = \emptyset$.

In the following, we will repeatedly make use of the following fact
about square-free graphs: if two distinct vertices have a common
neighbour, they must have a unique common neighbour, since a pair of
vertices with two common neighbours would form a $4$-cycle.

\begin{lemma}
\label{lem:Jcat-tech}
    Let $H$ be a square-free graph, let $k>0$ and let $P=v_0\dots
    v_k$ be a path in $H$.
    \begin{enumerate}
    \item For any $a\in\Gamma_H(v_0)-v_1$ and $\sigma \in
        \Homs{(J_P,y)}{(H,a)}$, $\sigma(u_j) = v_{j-1}$ for all
        $j\in[k-1]$.
    \item For any $b\in\Gamma_H(v_k)-v_{k-1}$ and $\sigma \in
        \Homs{(J_P,z)}{(H,b)}$, $\sigma(u_j) = v_{j+1}$ for all
        $j\in[k-1]$.
    \end{enumerate}
\end{lemma}
\begin{proof}
    The result is trivial for $k=1$ so we assume $k>1$.
    We prove the first part, by induction on~$j$.  The second part
    follows by symmetry (call the vertices on the path $v_k\dots
    v_0$ instead of $v_0\dots v_k$).

    First, take $j=1$. From the structure of~$J_P$, $\sigma(u_1)$ must
    be a neighbour of $\sigma(y)=a$ and of~$v_1$, which are distinct
    vertices.  $v_0$~is a common neighbour of $a$ and~$v_1$, so it must
    be their unique common neighbour, so $\sigma(u_1) = v_0$.
    Now, suppose that $\sigma(u_{j-1}) = v_{j-2}$.  As in the base case,
    $\sigma(u_j)$ must be some neighbour of $v_{j-2}$ and~$v_j$, which
    are distinct. $v_{j-1}$~is such a vertex, so it is the unique such
    vertex.
\end{proof}

\begin{lemma}\label{lem:J3-caterpillar}
    Let $H$ be a square-free graph. Let $k>0$ and let $P=v_0\dots
    v_k$ be a path in~$H$ with $\deg_H(v_j)$ odd for all
    $j\in\{1, \dots, k-1\}$.  Let $\Oy\subseteq\Gamma_H(v_0)$ and
    $\Oz\subseteq\Gamma_H(v_k)$, with $i=v_1\in \Oy$ and
    $s=v_{k-1}\in \Oz$.  For each $o\in\Oy-i$ and each $x\in\Oz-s$:
    \begin{enumerate}
    \item $|\Homs{(J_P,y,z)}{(H,o,x)}| = 0$,\label{cond:ox-cat}
    \item $|\Homs{(J_P,y,z)}{(H,o,s)}| = 1$,\label{cond:os-cat}
    \item $|\Homs{(J_P,y,z)}{(H,i,x)}| = 1$ and\label{cond:ix-cat}
    \item $|\Homs{(J_P,y,z)}{(H,i,s)}|$ is odd.\label{cond:is-cat}
    \end{enumerate}
\end{lemma}
\begin{proof}
    If $k=1$, $i=v_1$, $s=v_0$, $G(J_P)$ is the single edge $(y,z)$
    and $\tau(J_P)=\emptyset$.  For any $o\in \Oy-i$ and $x\in\Oy-s$,
    we have $(o,s), (i,s), (i,x)\in E(H)$ so $(o,x)\notin E(H)$
    because $H$~is square-free.  Parts 1--4 are immediate.
    For the remainder of the proof, we may assume that $k\geq 2$. 
    Note that when $k=2$, $i=s=v_1$ and this is the unique common
    neighbour of $v_0$ and~$v_2$ in~$H$.

    For part~1, suppose, towards a contradiction, that $\sigma \in
    \Homs{(J_P,y,z)}{(H,o,x)}$.  In particular, $\sigma \in
    \Homs{(J_P,y)}{(H,o)}$ so, by Lemma~\ref{lem:Jcat-tech}(1),
    $\sigma(u_1)=v_0$.  We also have $\sigma \in
    \Homs{(J_P,z)}{(H,x)}$ so, by Lemma~\ref{lem:Jcat-tech}(2),
    $\sigma(u_1)=v_2$.  But $P$~is a simple path so $v_0\neq v_2$.

    For part~2, let $\sigma\in \Homs{(J_P,y,z)}{(H,o,s)}$.  Since
    $\sigma\in \Homs{(J_P,y)}{(H,o)}$, $\sigma(u_j) = v_{j-1}$ for all
    $i\in[k-1]$ by Lemma~\ref{lem:Jcat-tech}(1).  But now, $\sigma$
    is completely determined, so it is the unique element of
    $\Homs{(J_P,y,z)}{(H,o,s)}$.  Part~3 follows similarly from
    Lemma~\ref{lem:Jcat-tech}(2).

    For part~4, first note that there is a
    homomorphism $\sigma^+\in \Homs{(J_P,y,z)}{(H,i,s)}$ with
    $\sigma^+(u_j) = v_{j+1}$ for all $j\in[k-1]$.  Now, for $m\in
    [k-1]$, let
    \begin{equation*}
        S_m = \{\sigma \in \Homs{(J_P,y,z)}{(H,i,s)} \mid m \text{ is
          minimal such that }\sigma(u_m)\neq v_{m+1}\}\,.
    \end{equation*}
    The sets $\{\sigma^+\}$ and $S_1, \dots, S_{k-1}$ partition
    $\Homs{(J_P,y,z)}{(H,i,s)}$.

    We claim that, for any $\sigma\in S_m$, $\sigma(u_j) = v_{j-1}$
    for all $j>m$.  This is trivial for~$S_{k-1}$ so let $\sigma\in
    S_m$ with $m<k-1$.  $\sigma(u_{m+1})$ must be a neighbour of
    both $\sigma(w_{m+1})=v_{m+1}$ and $\sigma(u_m)\in\Gamma_H(v_m)$.
    By definition of~$S_m$, these are distinct vertices so $v_m$~is
    their unique common neighbour and so $\sigma(u_{m+1}) = v_m$.
    Now, if $\sigma(u_j)=v_{j-1}$ for
    some $j\in \{m+1, \dots, k-2\}$, then $\sigma(u_{j+1})$ must be
    a neighbour of both $\sigma(w_{j+1}) = v_{j+1}$ and $v_{j-1}$: $v_j$~is
    the unique such vertex, so $\sigma(u_{j+1})=v_j$.  This
    establishes the claim.

    But, now, for any $\sigma\in S_m$, we have $\sigma(u_j) = v_{j+1}$
    for $j<m$ and $\sigma(u_j) = v_{j-1}$ for $j>m$.  $\sigma(y)=i$,
    $\sigma(z)=s$ and $\sigma(w_j)=v_j$ for each $j\in[k-1]$.
    Finally, $\sigma(u_m)$ may take any value in
    $\Gamma_H(v_m)-v_{m+1}$.  It follows that, for all~$m$, $|S_m| =
    \deg_H(v_m)-1$, which is even.  $|\Homs{(J_P,y,z)}{(H,i,s)}| = 1 +
    \sum_m|S_m|$, which is odd, as required.
\end{proof}

\subsection{Caterpillar gadgets and 4-cycles}
\label{sec:gadgets:squares}

Before proceeding to find hardness gadgets for square-free graphs in
the next section, we pause to show why $4$-cycles cause problems for
caterpillar gadgets and, in particular, why
Lemma~\ref{lem:J3-caterpillar} does not apply to graphs containing
$4$-cycles.

Consider first the one-edge caterpillar gadget~$J_1$ associated with
the path~$v_0v_1$ in the graph~$H_1$ in
Figure~\ref{fig:caterpillar-squares}.  This corresponds to $k=1$ in
Lemma~\ref{lem:J3-caterpillar} and we have $i=v_1$ and $s=v_0$.
Taking $\Oy = \Gamma_{\!H_1}(v_0) = \{v'_0,v_1\}$ and $\Oz =
\Gamma_{\!H_1}(v_1) = \{v_0,v'_1\}$ satisfies the conditions of the
lemma.  However, taking $o = v'_0 \in \Oy-i$ and $x = v'_1 \in \Oz -
s$, we have $|\Homs{(J_1,y,z)}{(H,o,x)}| = 1$ so
part~\ref{cond:ox-cat} of the lemma does not hold.  However, the other
three parts hold, as
\begin{align*}
    |\Homs{(J_1,y,z)}{(H,o,s)}|
        &= |\Homs{(J_1,y,z)}{(H,i,x)}| \\
        &= |\Homs{(J_1,y,z)}{(H,i,s)}|
        = 1\,.
\end{align*}

\begin{figure}[t]
\begin{center}
\begin{tikzpicture}[scale=1.2,node distance = 1.5cm]
\tikzstyle{dot}   =[fill=black, draw=black, circle, inner sep=0.15mm]
\tikzstyle{vertex}=[fill=black, draw=black, circle, inner sep=2pt]
\tikzstyle{dist}  =[fill=white, draw=black, circle, inner sep=2pt]
\tikzstyle{pinned}=[fill=white, draw=black, minimum size=9mm, circle,inner sep=0pt]

   \begin{scope}[shift={(-3,0)}]
       \node at (0,0.5) {$H_1$:};

       \node[vertex] (o) at (2,1) [label=315:$o$] {};
       \node[vertex] (s) at (2,0) [label= 45:$s$] [label=-90:$v_0$] {};
       \node[vertex] (i) at (3,0) [label=135:$i$] [label=-90:$v_1$] {};
       \node[vertex] (x) at (3,1) [label=225:$x$] {};

       \node at ($(o) + (210:0.4)$) {$v'_0$};
       \node at ($(x) + (330:0.4)$) {$v'_1$};

       \node[vertex] (a) at ($(x) + (45:1)$) {};
       \node[vertex] (b) at ($(o) + (150:1)$) {};
       \node[vertex] (c) at ($(b) + (-1,0)$) {};

       \draw (o) -- (s) -- (i) -- (x) -- (a);
       \draw (x) -- (o) -- (b) -- (c);
   \end{scope}

   \begin{scope}[shift={(-3,-2.5)}]
       \node at (0,1) {$J_1$:};

       \node[dist] (yy) at (2,1) [label=90:$y$] {};
       \node[dist] (zz) at (3,1) [label=90:$z$] {};
       \draw (yy) -- (zz);
   \end{scope}

   \begin{scope}[shift={(3,0)}]
       \node at (-0.75,0.5) {$H_k$:};

       \node[vertex] (v0)   at (0.0,0) [label=-90:$v_0$]                    {};
       \node[vertex] (v1)   at (1.0,0) [label=-90:$v_1$]    [label= 45:$i$] {};
       \node[vertex] (vk1)  at (2.5,0) [label=-90:$v_{k-1}$] [label=135:$s$] {};
       \node[vertex] (vk)   at (3.5,0) [label=-90:$v_k$]                    {};

       \node[vertex] (vd0)  at (0.0,1) [label=90:$v'_0$] [label=315:$o$] {};
       \node[vertex] (vd1)  at (1.0,1) [label=90:$v'_1$]                 {};
       \node[vertex] (vdk1) at (2.5,1) [label=90:$v'_{k-1}$]              {};
       \node[vertex] (vdk)  at (3.5,1)                   [label=225:$x$] {};
       \node at ($(vdk) + (330:0.4)$) {$v'_k$};

       \node[vertex] (d)   at ($(vdk) + (45:1)$) {};

       \draw (v0)  -- (vd0);
       \draw (v1)  -- (vd1);
       \draw (vk1) -- (vdk1);
       \draw (vk)  -- (vdk);

       \draw (v0)   -- (v1);
       \draw (vd0)  -- (vd1);
       \draw (vk1)  -- (vk);
       \draw (vdk1) -- (vdk) -- (d);

       % Dotted path from $v_1$ to $v_{k-1}$
       \node[dot] (dot2) at ($(v1)!0.5!(vk1)$) {};
       \node[dot] (dot1) at ($(dot2)!1mm!(v1)$) {};
       \node[dot] (dot3) at ($(dot2)!1mm!(vk1)$) {};
       \draw (v1) -- ($(v1)!0.75!(dot1)$);
       \draw (vk1) -- ($(vk1)!0.75!(dot3)$);

       % Dotted path from $v'_1$ to $v'_{k-1}$
       \node[dot] (dot5) at ($(vd1)!0.5!(vdk1)$) {};
       \node[dot] (dot4) at ($(dot5)!1mm!(vd1)$) {};
       \node[dot] (dot6) at ($(dot5)!1mm!(vdk1)$) {};
       \draw (vd1) -- ($(vd1)!0.75!(dot4)$);
       \draw (vdk1) -- ($(vdk1)!0.75!(dot6)$);
   \end{scope}

   \begin{scope}[shift={(3,-2.5)}]
       \node at (-0.75,1) {$J_P$:};
       \node[dist]    (y)   at (0.0,1) [label=90:$y$]      {};
       \node[vertex]  (u1)  at (1.0,1) [label=90:$u_1$]    {};
       \node[vertex]  (uk1) at (2.5,1) [label=90:$u_{k-1}$] {};
       \node[dist]    (z)   at (3.5,1) [label=90:$z$]      {};

       \node[pinned]  (w1)  at (1.0,0) {$v_1$};
       \node[pinned]  (wk1) at (2.5,0) {$v_{k-1}$};

       \draw (y) -- (u1) -- (w1);
       \draw (z) -- (uk1) -- (wk1);

       % Dotted path from $v'_1$ to $v'_{k-1}$
       \node[dot] (dot8) at ($(u1)!0.5!(uk1)$) {};
       \node[dot] (dot7) at ($(dot8)!1mm!(u1)$) {};
       \node[dot] (dot9) at ($(dot8)!1mm!(uk1)$) {};
       \draw (u1) -- ($(u1)!0.75!(dot7)$);
       \draw (uk1) -- ($(uk1)!0.75!(dot9)$);
   \end{scope}

\end{tikzpicture}
\end{center}
\caption{Examples of graphs containing $4$-cycles for which
  caterpillar gadgets (Definition~\ref{defn:caterpillar} and
  Lemma~\ref{lem:J3-caterpillar}) fail.  The graphs $H_1$ and~$H_k$
  ($k\geq2$) are shown, along with the caterpillar gadgets $J_1$
  and~$J_P$, corresponding to the paths $v_0v_1$ and $v_0\dots v_k$,
  respectively.  The labels $o$, $s$, $i$ and~$x$ are referenced in
  the text.}
\label{fig:caterpillar-squares}
\end{figure}
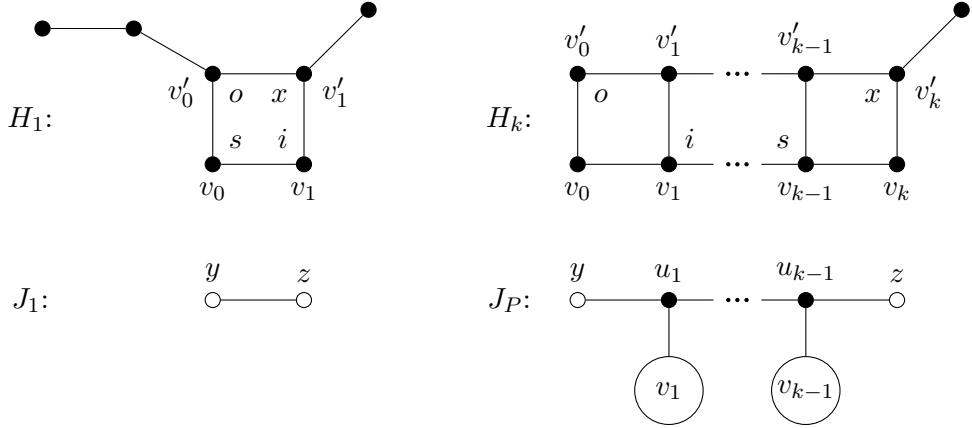

Now, consider longer paths such as the path $P=v_0\dots v_k$ in~$H_k$
in Figure~\ref{fig:caterpillar-squares}, for some $k\geq 2$.  The
associated caterpillar gadget~$J_P$ is also shown in the figure.  For
each $j\in \{1, \dots, k-1\}$, $\deg_{H_k}(v_i)$ is odd.  We have $i =
v_1$ and $s = v_{k-1}$ (with $i=s$ in the case $k=2$).    
Again, take $\Oy
= \Gamma_{\!H_k}(v_0) = \{v'_0,v_1\}$, take $\Oz = \Gamma_{\!H_k}(v_k)
= \{v_{k-1},v'_k\}$ and take $o = v'_0 \in \Oy-i$ and $x = 
v'_k
\in \Oz - s$.

Once again part~\ref{cond:ox-cat} of the lemma fails.  We have
$|\Homs{(J_P,y,z)}{(H_k,o,x)}| = 1$, since there is a homomorphism
that maps $u_j$ to~$v'_j$ for each $j\in\{1, \dots, k-1\}$.  This is
the only possible homomorphism from $(J_P,y,z)$ to~$(H_k,o,x)$ since
there is only one $k$-path from $o$ to~$x$ that the $k$-path in 
$J_P$
can be mapped to.  For a hardness gadget, it would suffice for
$|\Homs{(J_P,y,z)}{(H_k,o,x)}|$ to be even (not necessarily zero) but
it is odd for every~$k$.

For $H_k$, the other parts of the lemma fail, too.  We have
\begin{equation*}
    |\Homs{(J_P,y,z)}{(H,o,s)}|
        = |\Homs{(J_P,y,z)}{(H,i,x)}|
        = k.
\end{equation*}
When the target is~$(H,o,s)$,
the $k$-path in $J_P$ can be mapped to any of the $k$ $k$-paths in~$H_k$
from $o$ to~$s$ 
(following along $v'_0, v'_1,\cdots$ and then dropping down along an edge
$v'_j,v_j$ and then following $v_j,v_{j+1} \cdots v_{k-1}$). The  
case with target $(H,i,x)$ is similar.
So in both cases, the number of homomorphisms is~$k$.
When $k$~is odd,
this is not a real problem.  The purpose of
Lemma~\ref{lem:J3-caterpillar} is to show that caterpillar gadgets can
be used as~$J_3$ in a hardness gadget, and the definition of hardness
gadgets only requires that $|\Sos|$ and~$|\Six|$ (i.e.,
$|\Homs{(J_P,y,z)}{(H,o,s)}|$ and $|\Homs{(J_P,y,z)}{(H,i,x)}|$,
respectively) be odd and not necessarily~$1$.  However, this
relaxation doesn't help when $k$~is even.

Finally, for part~\ref{cond:is-cat}, consider a homomorphism from
$(J_P,y,z)$    
to~$(H,i,s)$.  
The image of the path $yu_1\dots u_{k-1}z$
in~$H$ must be a $k$-walk $v_1x_1\dots x_{k-1}v_{k-1}$ with the
property that $x_j$~is adjacent to~$v_j$ for each $j\in\{1, \dots,
k-1\}$.  This means that $x_j\in\{v_{j-1}, v'_j, v_{j+1}\}$.  There are
two kinds of $k$-walk satisfying these criteria.  The first kind uses
only the vertices $\{v_0, \dots, v_k\}$.  Such a walk must be either
$v_1v_0v_1 v_2\dots v_{k-1}$ or $v_1\dots v_\alpha v_{\alpha+1}
v_\alpha \dots v_{k-1}$ for some $\alpha\in\{1, \dots, k-1\}$.  The
second kind uses some of the vertices $\{v'_1, \dots, v'_{k-1}\}$.  Such
a walk must be of the form $v_1\dots v_\alpha v'_\alpha \dots v'_\beta
v_\beta \dots v_{k-1}$ for some $1\leq \alpha\leq \beta\leq k-1$.
There are $k$ walks of the first kind and $\tfrac12k(k-1)$ of the
second.  Thus,
\begin{equation*}
    |\Homs{(J_1,y,z)}{(H,i,s)}| = k + \tfrac12k(k-1) = \tfrac12k(k+1)\,,
\end{equation*}
which is odd if and only if $k$~is congruent to $1$ or~$2$, mod~$4$
but is required to be odd for all~$k$.

\begin{figure}[t]
\begin{center}
\begin{tikzpicture}[scale=1.2,node distance = 1.5cm]
\tikzstyle{dot}   =[fill=black, draw=black, circle, inner sep=0.15mm]
\tikzstyle{vertex}=[fill=black, draw=black, circle, inner sep=2pt]
\tikzstyle{dist}  =[fill=white, draw=black, circle, inner sep=2pt]
\tikzstyle{pinned}=[fill=white, draw=black, minimum size=9mm, circle,inner sep=0pt]

   \begin{scope}[shift={(-3,0)}]
       \node at (-0.75,0.5) {$H_k$:};

       \node[vertex] (v0)   at (0.0,0) [label=-90:{$\strut v_0$}]
                                                   [label=45:$o$] {};
       \node[vertex] (v1)   at (1.0,0) [label=-90:{$\strut v_1$}]
                                                   [label=45:$s$] {};
       \node[vertex] (v2)   at (2.0,0) [label=-90:{$\strut v_2$}] {};
       \node[vertex] (vk)   at (3.5,0) [label=-90:{$\strut v_k$}] {};

       \node[vertex] (vd0)  at (0.0,1) [label=90:{$\strut v'_0$}] {};
       \node[vertex] (vd1)  at (1.0,1) [label=90:{$\strut v'_1$}]
                                                  [label=-45:$i$] {};
       \node[vertex] (vd2)  at (2.0,1) [label=90:{$\strut v'_2$}]
                                                  [label=-45:$x$] {};
       \node[vertex] (vdk)  at (3.5,1) {};
       \node at ($(vdk) + (345:0.4)$) {$\strut v'_k$};

       \node[vertex] (d)   at ($(vdk) + (45:1)$) {};

       \draw (v0)  -- (vd0);
       \draw (v1)  -- (vd1);
       \draw (v2)  -- (vd2);
       \draw (vk)  -- (vdk);

       \draw (v0)  -- (v2);
       \draw (vd0) -- (vd2);
       \draw (vdk) -- (d);

       % Dotted path from $v_2$ to $v_k$
       \node[dot] (dot2) at ($(v2)!0.5!(vk)$) {};
       \node[dot] (dot1) at ($(dot2)!1mm!(v2)$) {};
       \node[dot] (dot3) at ($(dot2)!1mm!(vk)$) {};
       \draw (v2) -- ($(v2)!0.75!(dot1)$);
       \draw (vk) -- ($(vk)!0.75!(dot3)$);

       % Dotted path from $v'_2$ to $v'_k$
       \node[dot] (dot5) at ($(vd2)!0.5!(vdk)$) {};
       \node[dot] (dot4) at ($(dot5)!1mm!(vd2)$) {};
       \node[dot] (dot6) at ($(dot5)!1mm!(vdk)$) {};
       \draw (vd2) -- ($(vd2)!0.75!(dot4)$);
       \draw (vdk) -- ($(vdk)!0.75!(dot6)$);
   \end{scope}

   \begin{scope}[shift={(3,1.75)}]
       \node at (0,0) {$J_1$:};

       \node[pinned] (p1) at (1,0) {$\strut v'_0$};
       \node[dist]   (y)  at (2,0) [label=-90:$y$] {};
       \node[pinned] (p2) at (3,0) {$\strut v_1$};

       \draw (p1) -- (y) -- (p2);
   \end{scope}

   \begin{scope}[shift={(3,0.5)}]
       \node at (0,0) {$J_2$:};

       \node[pinned] (p3) at (1,0) {$\strut v'_1$};
       \node[dist]   (z)  at (2,0) [label=-90:$z$] {};
       \node[pinned] (p4) at (3,0) {$\strut v_2$};

       \draw (p3) -- (z) -- (p4);
   \end{scope}

   \begin{scope}[shift={(3,-0.75)}]
       \node at (0,0) {$J_3$:};

       \node[dist] (yy) at (1.5,0) [label=-90:$y$] {};
       \node[dist] (zz) at (2.5,0) [label=-90:$z$] {};

       \draw (yy) -- (zz);
   \end{scope}

\end{tikzpicture}
\end{center}
\caption{A hardness gadget for the graph~$H_k$ (see also
  Figure~\ref{fig:caterpillar-squares}).}
\label{fig:hardness-squares}
\end{figure}
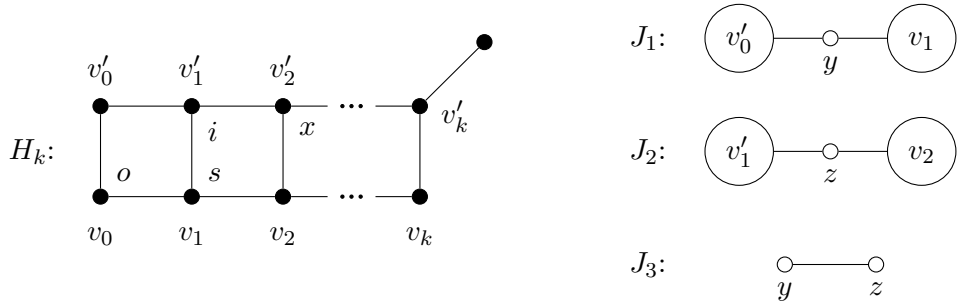

We note that \parhcol[H_1] is \parp{}-complete, as is \parhcol[H_k],
for every $k\geq 2$.  $H_1$~is an involution-free cactus graph with more
than one vertex so it is hard by the main theorem
of~\cite{GGR14:Cactus}.  We claim that $\calX = (i, s, (J_1,y),
(J_2,z), (J_3,y,z))$, as shown in Figure~\ref{fig:hardness-squares},
is a hardness gadget for~$H_k$.  We have $\Oy = \{v_0, v'_1\} =
\{o,i\}$ and $\Oz = \{v_1, v'_2\} = \{s,x\}$: both are even and
$i\in\Oy$ and $s\in\Oz$.  There is no edge~$ox$ in~$H_k$ so
$|\Sox|=0$, which is even.  There are edges $os$, $ix$ and~$is$
in~$H_k$, so $|\Sos| = |\Six| = |\Sis| = 1$, which is odd.  This
establishes that $\calX$ is a hardness gadget so, since $H_k$~is
involution-free, \parhcol[H_k] is \parp{}-complete by
Theorem~\ref{thm:hardness-gadget}.  Ironically, the part~$J_3$
of~$\calX$ is the one-edge caterpillar gadget associated with the
path~$v_1v'_1$ in~$H_k$.  The failure of
Lemma~\ref{lem:J3-caterpillar} in the presence of $4$-cycles only
means that caterpillar gadgets are not guaranteed to work, not that
they never work.

\section{Finding hardness gadgets}
\label{sec:finding-gadgets}

In this section, we show how to find hardness gadgets for all
connected, involution-free, square-free graphs.  The simplest case is when the
graph contains at least two vertices of even degree.  Faben and Jerrum
used the fact that all involution-free trees have at least two
even-degree vertices~\cite{FJ13}, though we use different gadgets
because we are dealing with graphs containing cycles as well as trees.
For graphs with
only one even-degree vertex, we show that an appropriate vertex
deletion produces a component with more than one even-degree vertex
and show how to simulate such a vertex deletion using gadgets.

This leaves graphs where every vertex has odd degree.  In
Section~\ref{sec:odd-cycles}, we show how to use odd-length cycles to
find a hardness gadget.  The remaining case, bipartite graphs in which
every vertex has odd degree, is covered in
Section~\ref{sec:gadget:bipartite}, where we use
Lemma~\ref{lem:Lovasz}, our version of \Lovasz{}'s result, to
non-constructively demonstrate that a hardness gadget always exists.

We will use the following fact.

\begin{lemma}
\label{lem:invo-free-cycle}
    An involution-free graph with at least two vertices but at most
    one even-degree vertex contains a cycle.
\end{lemma}
\begin{proof}
    We prove the contrapositive.  Let $G$ be an involution-free
    acyclic graph.  At most one component of $G$~is an isolated vertex
    so, if $G$~has two or more vertices, it has at least one component
    with two or more vertices. 
    This component is an
    involution-free tree which, by \cite[Lemma~5.3]{FJ13}, contains at
    least two vertices of even degree.
\end{proof}

\subsection{Even-degree vertices}\label{sec:even-degree}

We prove that involution-free graphs containing at
least one vertex of positive, even degree have a hardness gadget.  In
this section, we will use one extra kind of gadget.

\begin{definition}
    For a vertex~$v\in V(H)$, define the \emph{neighbourhood gadget}
    $J_{\Gamma(v),x} = (G, \{w\mapsto v\})$, where $G$~is the single edge
    $(x,w)$.
\end{definition}

It is immediate from the definition that, for any $v\in V(H)$,
\begin{equation*}
    |\Homs{(J_{\Gamma(v),x},x)}{(H,u)}| = \begin{cases}
                                     \ 1 & \text{if $u\in\Gamma_H(v)$} \\
                                     \ 0 & \text{otherwise.}
                                     \end{cases}
\end{equation*}

We first show how to find hardness
gadgets for connected graphs containing at least two even-degree
vertices (their degree must be positive, since the graph is connected)
and then deal with the harder case of graphs containing exactly one
vertex of positive, even degree.  The following lemma constructs a
caterpillar gadget, so the lemma depends
on $H$~being square-free.  The extended
conclusion about pinned vertices is needed for technical reasons in
the proof of Lemma~\ref{lem:even-deg}.

\begin{lemma}\label{lem:two-even}
    Let $H$ be a connected, square-free graph with at least two
    even-degree vertices. Then $H$ has a hardness gadget $(i, s, (J_1,
    y), (J_2, z), (J_3, y, z))$.  Furthermore, we can choose $J_1$,
    $J_2$ and~$J_3$ so that each contains at least one pinned vertex.
\end{lemma}
\begin{proof}
    Let $v_0\dots v_m$ be a path in~$H$ between distinct even-degree
    vertices $v_0$ and~$v_m$ and let $P=v_0\dots v_k$, where
    $k\in\{1, \dots, m\}$ is minimal such that $\deg_H(v_k)$ is
    even.  We claim that $(v_1, v_{k-1}, (J_{\Gamma(v_0),y}, y),
    (J_{\Gamma(v_k),z}, z), (J_P,y,z))$ is a hardness gadget.
    $|\Oy|$ and~$|\Oz|$ are even because $v_0$ and~$v_k$ have even
    degree; and they contain $v_1$ and~$v_{k-1}$, respectively.
    The remaining properties required by
    Definition~\ref{defn:hardness-gadget} hold by
    Lemma~\ref{lem:J3-caterpillar}, since $v_1, \dots, v_{k-1}$ have
    odd degree.

    Each of $J_{\Gamma(v_0),y}$ and~$J_{\Gamma(v_k),z}$ contains a
    pinned vertex and, if $k>1$, $J_P$~also contains at least one
    pinned vertex.  If $k=1$, then $G(J_P)$ is the single edge
    $(y,z)$ and $\tau(J_P) = \emptyset$.  However, we may add to
    $G(J_P)$ 
a new vertex~$w_0$ and
    an edge $(w_0,y)$ and set $\tau(J_P) = \{w_0\mapsto
    v_0\}$: this requires $y$~to be mapped to a neighbour of~$v_0$.
    This has no effect on the hardness gadget since
    Definition~\ref{defn:hardness-gadget} only imposes requirements on
    $|\Homs{(J_3,y,z)}{(H,a,b)}|$ when $a\in \Oy$.
    Since $\Oy = \Gamma_H(v_0)$, we are already only considering
    homomorphisms that map $y$ to a neighbour of~$v_0$ and the change
    to~$J_3$ is merely restating this condition.
\end{proof}

It is worth noting that, since all involution-free trees have at least
two even-degree vertices, Lemma~\ref{lem:two-even} implies Faben and
Jerrum's dichotomy for \parhcol{} where $H$~is a tree~\cite{FJ13}.
They also use two even-degree vertices but their gadgets rely on the
fact that there is a unique path between two vertices of a tree, which
doesn't hold in general graphs.  However, from
Lemma~\ref{lem:two-even}, we conclude that uniqueness of the path is not
required and we can prove hardness even when there are multiple paths
between even-degree vertices.

To handle graphs with fewer than two vertices of even degree, we first
investigate the results of deleting vertices from such
graphs.  If we delete the unique even-degree vertex from a connected
graph, then each component of the resulting graph contains at least
one vertex of even degree.  If we are lucky, one of the resulting
components will contain two or more vertices of even degree, raising
the hope that we can use Lemma~\ref{lem:two-even} to prove
\parp{}-completeness.  If all of the resulting components have exactly
one even-degree vertex, then we can iterate, deleting those vertices
to obtain yet more fragments.  As long as the graph contains at least
one cycle, it is not hard to see that we can eventually obtain a
component with two or more even-degree vertices.  However, to apply
Lemma~\ref{lem:two-even}, we must ensure that the resulting component
has no involution.  We prove this in the following two lemmas.

\begin{lemma}
\label{lem:one-even-asym}
    Let $H$ be an involution-free graph with exactly one vertex~$v$ of
    positive, even degree.  Then $H'=H-v$ is also involution-free.
\end{lemma}
\begin{proof}
    Each vertex $u\in \Gamma_H(v)$ has odd degree in~$H$ and has
    exactly one neighbour removed, so $\deg_{H'}(u)$ is even.
    Suppose, towards a contradiction, that $\rho$~is an involution
    of~$H'\!$.  No automorphism can map an odd-degree vertex to an
    even-degree vertex or vice-versa and $\Gamma_H(v)$~is exactly the
    set of even-degree vertices in $H'\!$.  Therefore, the restriction
    of~$\rho$ to the neighbours of~$v$ is a permutation.
    Define $\hat\rho\colon V(H)\to V(H)$ by
    $\hat\rho(v)=v$ and $\hat\rho(w)=\rho(w)$ for $w\neq v$.
    $\hat\rho$ preserves all edges in~$H'$ and all edges incident
    on~$v$ in~$H$, so it is an involution of~$H$, contradicting the
    supposition that $H$~has no involution.
\end{proof}

So far, we have described our goal as being to iteratively delete
vertices until we find a component with more than one even-degree
vertex.  This is a useful intuition but we do not know how to simulate
such a sequence of vertex deletions using gadgets.  Instead, we show
how to achieve the goal of a component with more than one even-degree
vertex by deleting a set of vertices, which we do know how to do with
a gadget.

For a vertex $v\in V(H)$ and an integer~$r\geq 0$, let
$B_r(v) = \{u\in V(H) \mid \dist(u,v) = r\}$.

\begin{corollary}
\label{cor:one-even-asym}
    Let $H$ be an involution-free graph that has
    exactly one vertex~$v$ of positive, even degree.  For some~$r$,
    $H-B_r(v)$ has an involution-free component~$H^*$ that does not
    contain~$v$ but does contain at least two even-degree vertices.
    Furthermore, we can take~$r = \min\, \{\dist(v,w)\mid w \text{ is
      on a cycle}\}$.
\end{corollary}
\begin{proof}
    $H$~contains a cycle by Lemma~\ref{lem:invo-free-cycle} so we can
    take~$r$ as in the statement of the lemma and this is well-defined.
    If $r=0$, then $v$~is in some cycle~$C$ in~$H$.  $H-v$~has no
    involution by Lemma~\ref{lem:one-even-asym}, so no component of
    $H-v$ has an involution.  The component~$H^*$ of $H-v$ that contains
    $C-v$ contains at least two vertices of $\Gamma_H(v)$ ($v$'s two
    neighbours in~$C$) and these vertices have even degree
    in~$H^*\!$.  $H^*$~does not, of course, contain~$v$.

    Suppose that $r>0$.  By the choice of~$r$, there must be a
    component~$H'$ of $H-B_{r-1}(v)$ that contains a
    vertex~$v_r\in B_r(v)$ that is in a cycle~$C'$ of~$H'$.
    Since no vertex at distance less than~$r$ from~$v$ is in a cycle
    in~$H$, there is
    a unique path from $v$ to~$v_r$.  Let this be $v_0\dots
    v_r$, where $v=v_0$.  A simple induction on $j=0, \dots,
    r-1$, using Lemma~\ref{lem:one-even-asym}, shows that the
    component of $H-v_j$ containing~$v_r$ has no involution, does not
    contain~$v$ and has exactly one even-degree vertex: namely,
    $v_{j+1}$.  In particular, the component of $H-v_{r-1}$ that
    contains~$v_r$ is~$H'$.  But, now, the component of $H'-v_r$
    that contains $C'-v_r$ has no involution (because no component
    of $H'-v_r$ has an involution) and contains at least two vertices
    of even degree (because $v_r$ has at least two neighbours
    in~$C'$).  Further, this component is the component~$H^*$ of
    $H-B_r(v)$ that we seek.
\end{proof}

Thus, starting with an involution-free graph~$H$ containing only one
vertex of positive, even degree, we have shown how to make a set of
vertex deletions (some set $B_r(v)$) to obtain an involution-free
component~$H^*$ with at least two even-degree vertices.  We now show
that we can achieve these vertex deletions using gadgetry.
The following technical lemma allows us to construct a gadget that, in
a sense, ``selects'' the vertices of~$H^*$ within~$H$.

\begin{lemma}
\label{lem:Jpath}
    Let $H$ be a graph, let $P=x_0\dots x_{r+1}$ with $r\geq 0$ be a
    path in~$H$ and let $w\in V(H)$.  If
    every vertex in~$H$ within distance $r-1$ of~$w$ has odd
    degree, then $|\Homs{(P,x_0)}{(H,w)}|$ has opposite parity to the
    number of distinct $r$-paths in~$H$ from $w$ to vertices of
    even degree.
\end{lemma}
\begin{proof}
    We prove the lemma by induction on~$r$.
    For $r=0$, the result is trivial.  The condition on vertices
    within distance $r-1$ is vacuous.  The number of $0$-paths
    from $w$ to vertices of even degree is zero if $\deg(w)$ is odd;
    it is one if $\deg(w)$ is even; and $|\Homs{(P,x_0)}{(H,w)}| =
    \deg(w)$.

    Suppose the result holds for the path $P=x_0\dots x_{r+1}$ and consider the
    path $Px_{r+2}$ and a graph~$H$ in which every vertex within
    distance~$r$ of~$w$ has odd degree.

    Every homomorphism~$\sigma$ from $(Px_{r+2},x_0)$ to~$(H,w)$
    induces a homomorphism~$\hat\sigma$ from $(P,x_0)$ to $(H,w)$.
    Write $\sigma \sim \sigma'$ if $\hat\sigma = \hat\sigma'\!$.
    $\sim$~is an equivalence relation and its equivalence classes
    partition $\Homs{(Px_{r+2},x_0)}{(H,w)}$.  Let $\eqclass{\sigma}$ be
    the $\sim$-equivalence class of~$\sigma$.

    If every vertex within distance~$r$ of~$w$ in~$H$ has odd
    degree, there are no $r$-paths from $w$ to vertices of even
    degree so, by the inductive hypothesis, there are an odd number of
    homomorphisms from $(P,x_0)$ to~$(H,w)$, so there are an odd number of
    equivalence classes.  
Further, $|\eqclass{\sigma}| = \deg(\sigma(x_{r+1}))$ 
(this is well-defined since
    $\sigma(x_{r+1})=\hat\sigma(x_{r+1})$, so all homomorphisms
    $\sigma'\in\eqclass{\sigma}$ agree on the value of $\sigma'(x_{r+1})$).
    Any vertex of even degree is at distance at least $r+1$
    from~$w=\sigma(x_0)$ so, if $\deg_H(\sigma(x_{r+1}))$ is even, then the
    $r$-walk $\sigma(x_0) \sigma(x_1)\dots \sigma(x_{r+1})$ is, in fact, a
    simple $(r+1)$-path.  Therefore, the number~$N$ of
    even-cardinality equivalence classes is equal to the number of
    $(r+1)$-paths in~$H$ from~$w$ to a vertex of even degree, and
    subtracting these from the total number of equivalence classes gives 
    $|\Homs{(Px_{r+2},x_0)}{(H,w)}| \equiv 1-N \bmod2$, as required.
\end{proof}

Now, we can obtain a hardness gadget for~$H$ by combining the
``selection gadget'' with the hardness gadget for the subgraph~$H^*$
given to us by Corollary~\ref{cor:one-even-asym}.

\begin{lemma}\label{lem:even-deg}
    Any involution-free, square-free graph~$H$ that has exactly one
    vertex~$v$ of positive, even degree has a hardness gadget.
\end{lemma}
\begin{proof}
    Let $r = \min\,\{\dist(v,w)\mid w\text{ is on a cycle}\}$.
    By Corollary~\ref{cor:one-even-asym}, there is an
    involution-free component~$H^*$ of $H-B_r(v)$ that does not
    contain~$v$ but contains at least two vertices of even degree.
    $H^*$~is square-free because it is an induced subgraph of a
    square-free graph.
    Therefore, by Lemma~\ref{lem:two-even}, $H^*$~has a
    hardness gadget $\calX^* = (i, s, (J^*_1,y), (J^*_2,z), (J^*_3,y,z))$
    in which each of $J_1^*$, $J_2^*$ and~$J_3^*$ contains a pinned
    vertex.

    We construct a hardness gadget~$\calX$ for~$H$ from~$\calX^*\!$.  Let
    $P$~be a path of length $r+1\geq 1$, with vertices $x_0\dots
    x_{r+1}$.  Let $J_1=(G,\tau)$ be the partially $H$-labelled graph
    such that $\tau=\tau(J^*_1)$ and $G$~is defined from $G(J^*_1)$ as
    follows: start with $G(J^*_1)$ and, for every vertex $u\in
    G(J^*_1)$, add a new copy of~$P$ and identify that copy's
    vertex~$x_0$ with~$u$.  Define $J_2$ and~$J_3$ similarly, from
    $J_2^*$ and~$J_3^*$.  We claim that the tuple
    \begin{equation*}
        \calX = \big(i, s, (J_1, y), (J_2, z), (J_3, y, z)\big)
    \end{equation*}
    is the desired hardness gadget for~$H$.

    To find out what $\calX$ does, we first consider homomorphisms from
    one copy of the path~$P$ to~$H$.  For a vertex $w\in V(H)$, let
    $N_w = |\Homs{(P,x_0)}{(H,w)}|$.  
    If $\dist(v,w)=r$ (i.e., $w\in
    B_r(v)$), then there is a unique $r$-path from $w$ to a vertex of
even degree.  This is because $v$~is the unique vertex of even degree 
and, if there were distinct $r$-paths $Q_1$ and $Q_2$ from $w$ to~$v$ then $Q_1\cup Q_2$ would contain a cycle, which would contain vertices at distance strictly less than~$r$ from~$v$, contradicting the definition of~$r$.
    If $\dist(v,w)>r$, then there are
    no $r$-paths from $w$ to even-degree vertices.  Therefore, by
    Lemma~\ref{lem:Jpath}, $N_w$~is even if $\dist(v,w)=r$ and
    $N_w$~is odd if $\dist(v,w)>r$ (we will see that the parity
    of~$N_w$ does not matter if $\dist(v,w)<r$).

    Now, let $a\in V(H)$ and consider homomorphisms $\sigma, \sigma'\in
    \Homs{(J_1,y)}{(H,a)}$.  Write $\sigma\sim\sigma'$ if
    $\sigma(u)=\sigma'(u)$ for all $u\in V(G(J^*_1))$ and write
    $\eqclass{\sigma}$ for the $\sim$-equivalence class
    containing~$\sigma$.  $|\Homs{(J_1,y)}{(H,a)}|$ is the sum of
    the sizes of the $\sim$-equivalence classes.  For any~$\sigma$, we
    have
    \begin{equation*}
        |\eqclass{\sigma}| \ \ = \!\!\!\!
            \prod_{x\in V(G(J^*_1))}\!\!\!\! |\Homs{(P, x_0)}{(H, \sigma(x))}|\,.
    \end{equation*}
    Therefore, $|\eqclass{\sigma}|$ is even if $\sigma$~maps any
    vertex of $G(J^*_1)$ into $B_r(v)$.  In this case,
    $|\eqclass{\sigma}|$ contributes nothing to the sum, modulo~$2$.

    Thus, we may restrict our attention to homomorphisms from $J_1^*$
    to~$H$ that have no vertex in $B_r(v)$ in their image.  $J_1^*$~is
    connected and contains a vertex pinned to a vertex in~$H^*\!$.
    Therefore, restricting to homomorphisms that have no vertex in
    $B_r(v)$ in their image means restricting to homomorphisms whose
    image is wholly within~$H^*\!$.  For any vertex $w\in H^*\!$,
    $\dist_H(v,w)>r$, so this gives
    \begin{equation*}
        |\Homs{(J_1,y)}{(H,a)}|
            \equiv |\Homs{(J_1^*,y)}{(H^*\!,a)}| \pmod2\,,
    \end{equation*}
    for any $a\in V(H^*)$ and $|\Homs{(J_1,y)}{(H,a)}|\equiv 0
    \bmod2$, for $a\notin V(H^*)$; and similarly for $J_2$ and~$J_3$.
Thus, since $\calX^*$ is a hardness gadget for~$H^*$,
$\calX$ is a hardness gadget for~$H$.
\end{proof}

The proof of Lemma~\ref{lem:even-deg} does not explicitly use
caterpillar gadgets.  However, the hardness gadget~$\calX$ is
constructed from~$\calX^*\!$, which was produced by
Lemma~\ref{lem:two-even}.  It follows that $J_3^*$~is a caterpillar
gadget, so Lemma~\ref{lem:even-deg} requires $H$ to be square-free, as
stated.

\subsection{Odd cycles}
\label{sec:odd-cycles}

In the previous section, we showed how to find a hardness gadget for
any involution-free, square-free graph containing at least one vertex
of even degree.  In this
section, we show that any square-free graph in which all vertices have
odd degree has a hardness gadget if it has an odd cycle.
We first introduce a gadget for selecting certain vertices in cycles.

\begin{definition}\label{defn:Jcycle}
    (See Figure~\ref{fig:cycle}).
    Let $P=v_1\dots v_k$ be a path in~$H$. For any $\ell > \max\,
    \{2,k\}$, define the \emph{$\ell$-cycle gadget} $J_{\ell,P,x} =
    (G,\tau)$ where $G$~is the cycle $xu_1\dots u_{\ell-1}x$ and
    $\tau = \{u_1\mapsto v_1, \dots, u_k\mapsto v_k\}$.
\end{definition}

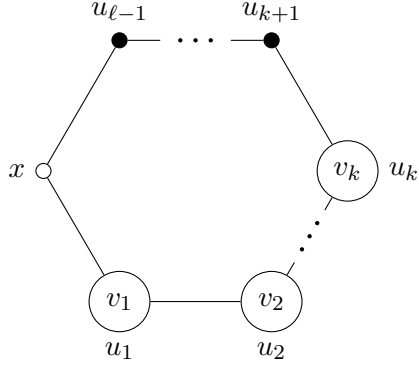
\begin{figure}
\begin{center}
\begin{tikzpicture}[scale=2,node distance = 1.5cm]
\tikzstyle{dot}   =[fill=black, draw=black, circle, inner sep=0.15mm]
\tikzstyle{vertex}=[fill=black, draw=black, circle, inner sep=2pt]
\tikzstyle{dist}  =[fill=white, draw=black, circle, inner sep=2pt]
\tikzstyle{pinned}=[draw=black, minimum size=8mm, circle]

    % Vertices
    \node[dist]   (x)  at (180:1) [label=180:$x$]   {};
    \node[vertex] (u5) at (120:1)[label=90:$u_{\ell-1}$] {};
    \node[vertex] (u4) at (60:1) [label=90:$u_{k+1}$] {};
    \node[pinned] (u1) at (240:1) [label=-90:$u_1$] {$v_1$};
    \node[pinned] (u2) at (300:1) [label=-90:$u_2$] {$v_2$};
    \node[pinned] (u3) at (  0:1) [label=  0:$u_k$] {$v_k$};

    % Solid edges
    \draw (u5) -- (x) -- (u1) -- (u2);
    \draw (u3) -- (u4);

    % Dotted edge from u_2 to u_3
    \node[dot] (b) at ($(u2)!0.5!(u3)$) {};
    \node[dot] (a) at ($(b)!1mm!(u2)$) {};
    \node[dot] (c) at ($(b)!1mm!(u3)$) {};
    \draw (u2) -- ($(u2)!0.75!(a)$);
    \draw (u3) -- ($(u3)!0.75!(c)$);

    % Dotted edge from u_4 to u_5
    \node[dot] (e) at ($(u4)!0.5!(u5)$) {};
    \node[dot] (d) at ($(e)!1mm!(u4)$) {};
    \node[dot] (f) at ($(e)!1mm!(u5)$) {};
    \draw (u4) -- ($(u4)!0.67!(d)$);
    \draw (u5) -- ($(u5)!0.67!(f)$);

\end{tikzpicture}
\end{center}
\caption{The $\ell$-cycle gadget $J_{\ell,P,x}$ corresponding to a
  path $P=v_1\dots v_k$ in an $\ell$-cycle in~$H$.}
\label{fig:cycle}
\end{figure}

Recall that the odd-girth of a graph is the length of its shortest odd
cycle.  By convention, the odd-girth of a graph without odd
cycles is infinite; in the following, we write ``a graph whose
odd-girth is~$\ell$'' as a short-hand for ``a graph whose odd-girth is
finite and equal to~$\ell$.''

\begin{lemma}\label{lem:cycle-to-cycle}
    Let $H$~be a graph whose odd-girth is~$\ell$ and let $G$~be an
    $\ell$-cycle.  The image of~$G$ under any homomorphism from $G$
    to~$H$ is an $\ell$-cycle in~$H$.
\end{lemma}
\begin{proof}
    Let $G = u_0\dots u_{\ell-1}u_0$.  Since $G$~is an $\ell$-cycle
    and $H$~contains an $\ell$-cycle, $\Homs{G}{H}$ is non-empty so
    let $\sigma\in\Homs{G}{H}$.  Let $C$~be the image of~$G$
    under~$\sigma$, i.e., subgraph of~$H$ consisting of vertices
    $\{\sigma(u_0), \dots, \sigma(u_{\ell-1})\}$ and edges
    $\{(\sigma(u_j),\sigma(u_{j+1})) \mid 0\leq j<\ell\}$, with
    addition on indices carried out modulo~$\ell$.  Suppose towards a
    contradiction that $C$~is not an $\ell$-cycle.  Since $C$~has at
    most $\ell$~vertices and at most $\ell$~edges, it cannot have an
    $\ell$-cycle as a proper subgraph. Since $H$~has no odd cycles
    shorter than~$\ell$, $C$~must be bipartite.  But then the walk
    $\sigma(u_0) \sigma(u_1) \dots \sigma(u_{\ell-1}) \sigma(u_0)$ is an
    odd-length walk from a vertex to itself and no such walk can exist
    in a bipartite graph.
\end{proof}

\begin{corollary}\label{cor:Jcycle}
    Let $H$ be a graph whose odd-girth is~$\ell$.  For any path~$P$ on
    fewer than $\ell$~vertices, $|\Homs{(J_{\ell,P,x},x)}{(H,v)}|$ is
    the number of $\ell$-cycles in~$H$ that contain the path~$vP$.
\end{corollary}
\begin{proof}
    By Lemma~\ref{lem:cycle-to-cycle}, the image of $G(J_{\ell,P,x})$ under
    any homomorphism to~$H$ is an $\ell$-cycle in~$H$ and, because of
    the pinning and distinguished vertex, this cycle must contain the
    path~$vP$.
\end{proof}

Let $\NCl(vw)$ be the number of $\ell$-cycles in~$H$ containing the
edge $(v,w)$. 

\begin{lemma}\label{lem:RLv-even}
    Let $H$ be a graph whose odd-girth is~$\ell$.  Every vertex $v\in
    V(H)$ has an even number of neighbours~$w$ such that $\NCl(vw)$ is
    odd.
\end{lemma}
\begin{proof}
    If $v$~is not in any $\ell$-cycle, the claim is vacuous: the even
    number is zero.  Otherwise, let $C=vw_1\dots w_{\ell-1}v$ be an
    $\ell$-cycle in~$H$.  If $w_j\in\Gamma_H(v)$ for some even $j\neq
    \ell-1$, the odd cycle $vw_1\dots w_jv$ contradicts the stated
    odd-girth of~$H$.  If $w_j\in\Gamma_H(v)$ for some odd $j\neq 1$,
    the odd cycle $vw_j\dots w_{\ell-1}v$ contradicts the odd-girth.
    Therefore, $w_1$ and~$w_{\ell-1}$ are the only vertices in~$C$
    that are adjacent to~$v$ and every $\ell$-cycle
    through~$v$ contributes exactly~$2$ to $\sum_{w\in\Gamma_H(v)}
    \NCl(vw)$.  Therefore, the sum is even, so it has an even number
    of odd terms.
\end{proof}

\begin{lemma}\label{lem:edge-in-odd}
    Let $H$ be a square-free graph whose odd-girth is~$\ell$.  If
    $H$~contains an edge that is in an odd number of $\ell$-cycles,
    then $H$~has a hardness gadget.
\end{lemma}

Note that, for the case $\ell=3$, any edge in a $3$-cycle in~$H$ must be in
exactly one $3$-cycle since, if an edge $(x,y)$ is in distinct
$3$-cycles $xyzx$ and $xyz'x$, then $xzyz'x$ is a $4$-cycle in~$H$,
which is forbidden by the hypothesis of the lemma.  The absence of
$4$-cycles is also required for the caterpillar gadget produced in the
proof.

\begin{proof}
    Let $(i, s)$ be an edge in an odd number of $\ell$-cycles in~$H$.
    Let $J_1$~be the $\ell$-cycle gadget $J_{\ell,s,y}$ (so $\tau(J_1) =
    \{u_1\mapsto s\}$) and let $J_2$~be the $\ell$-cycle gadget
    $J_{\ell,i,z}$.  Let $G(J_3)$~be the single edge $(y,z)$ and let
    $\tau(J_3)=\emptyset$ ($J_3$~is, technically, a caterpillar gadget
    but it is easier to analyse it directly).

    We claim that $(i, s, (J_1,y), (J_2,z), (J_3,y,z))$ is a hardness
    gadget for~$H$.  By Corollary~\ref{cor:Jcycle},
    $|\Homs{(J_{\ell,s,y},y)}{(H,v)}|$ is the number of $\ell$-cycles
    in~$H$ that contain the edge $(v,s)$, so
    \begin{equation*}
        \Oy = \{v\in V(H)\mid (v,s) \text{ is in an odd number of
                                                      $\ell$-cycles}\}\,.
    \end{equation*}
    Thus, $|\Oy|$ is even by Lemma~\ref{lem:RLv-even}. $\Oy$~contains $i$
    by the choice of the edge $(i,s)$ in an odd number of
    $\ell$-cycles.  Similarly, $\Oz$~is even and contains~$s$.
    To verify the remaining properties required by
    Definition~\ref{defn:hardness-gadget}, note that $J_3$~is a single
    edge so, for any $a,b\in V(H)$, $|\Homs{(J_3,y,z)}{(H,a,b)}|$
    is~$1$ if $(a,b)\in E(H)$ and~$0$, otherwise.  We have
    $\Oy\subseteq \Gamma_{\!H}(s)$ and $\Oz\subseteq \Gamma_{\!H}(i)$
    so, for any $o\in \Oy-i$ and any $x\in \Oz-s$, $H$~contains the
    edges $(o,s)$, $(s,i)$ and $(i,x)$ but it cannot contain the
    edge~$(o,x)$ because $H$~is square-free.
\end{proof}

\begin{figure}[t]
\begin{center}
\begin{tikzpicture}[scale=1,node distance = 1.5cm]
\tikzstyle{dot}   =[fill=black, draw=black, circle, inner sep=0.15mm]
\tikzstyle{vertex}=[fill=black, draw=black, circle, inner sep=2pt]
\tikzstyle{dist}  =[fill=white, draw=black, circle, inner sep=2pt]
\tikzstyle{pinned}=[fill=white, draw=black, minimum size=9mm, circle,inner sep=0pt]

\newcommand{\drawH}[1]
{
    \foreach \x in {0, 40, ..., 360}
    {
        \node[draw=#1,fill=#1,circle,inner sep=2pt] at (\x:1.5) {};
    }
    \draw[#1] (0,0) circle (1.5);

    \node[text=#1] at (  0:1.1) {$v_k$};
    \node[text=#1] at ( 40:1  ) {$v_{k-1}$};
    \node[text=#1] at (160:1.1) {$v_1$};
    \node[text=#1] at (200:1.1) {$v_0$};
    \node[text=#1] at (240:1.1) {$v_{\ell-1}$};
    \node[text=#1] at (320:1  ) {$v_{k+1}$};

    \draw [#1,<-] (200:1.8) arc (-160:0:1.8);
    \node[text=#1,fill=white,draw=white,circle, inner sep=0.5pt] at (300:1.8) {$P$};
}

    %
    % J_1
    % 
    \begin{scope}[shift={(-4,0)}]
        \drawH{gray!80!white};

        \node at (140:3.5) {$J_1$:};

        \draw (0,0) circle (2.5);

        \node[pinned] at (  0:2.5) {$v_k$};
        \node[vertex] at ( 40:2.5) {};
        \node[vertex] at ( 80:2.5) {};
        \node[vertex] at (120:2.5) {};
        \node[dist]   at (160:2.5) {};
        \node[pinned] at (200:2.5) {$v_0$};
        \node[pinned] at (240:2.5) {$v_{\ell-1}$};
        \node[pinned] at (280:2.5) {};
        \node[pinned] at (320:2.5) {$v_{k+1}$};

        \node at (160:2.85) {$y$};
    \end{scope}

    %
    % J_2
    % 
    \begin{scope}[shift={(4,0)}]
        \drawH{gray!80!white};

        \node at (140:3.5) {$J_2$:};

        \draw (0,0) circle (2.5);

        \node[pinned] at (  0:2.5) {$v_k$};
        \node[dist]   at ( 40:2.5) {};
        \node[vertex] at ( 80:2.5) {};
        \node[vertex] at (120:2.5) {};
        \node[vertex] at (160:2.5) {};
        \node[pinned] at (200:2.5) {$v_0$};
        \node[pinned] at (240:2.5) {$v_{\ell-1}$};
        \node[pinned] at (280:2.5) {};
        \node[pinned] at (320:2.5) {$v_{k+1}$};

        \node at (40:2.85) {$z$};
    \end{scope}

    %
    % J_3
    %
    \begin{scope}[shift={(0,-6)}]
        \drawH{gray!80!white};

        \node at (160:4.5) {$J_3$:};

        \draw (0:3.3) arc (0:200:3.3);
        \node[dist]        at (  0:3.3) {};
        \node[vertex] (u4) at ( 40:3.3) {};
        \node[vertex] (u3) at ( 80:3.3) {};
        \node[vertex] (u2) at (120:3.3) {};
        \node[vertex] (u1) at (160:3.3) {};
        \node[dist]        at (200:3.3) {};

        \node at (200:3.65) {$y$};
        \node at (  0:3.65) {$z$};

        \node[pinned] (w4) at ( 40:2.3) {$v_{k-1}$};
        \node[pinned] (w3) at ( 80:2.3) {};
        \node[pinned] (w2) at (120:2.3) {};
        \node[pinned] (w1) at (160:2.3) {$v_1$};

        \foreach \x in {1,2,3,4}
            \draw (u\x) -- (w\x);
    \end{scope}

\end{tikzpicture}
\end{center}
\caption{The parts $J_1$, $J_2$ and $J_3$ of the hardness gadget
    constructed in the proof of Lemma~\ref{lem:odd-cycle}.  The
    corresponding cycle in~$H$ is indicated in grey within each
    gadget.  The path~$P = v_k\dots v_{\ell-1}v_0$ is undirected but
    the arrow indicates the order in which the vertices are listed.}
\label{fig:gadget-from-cycle}
\end{figure}
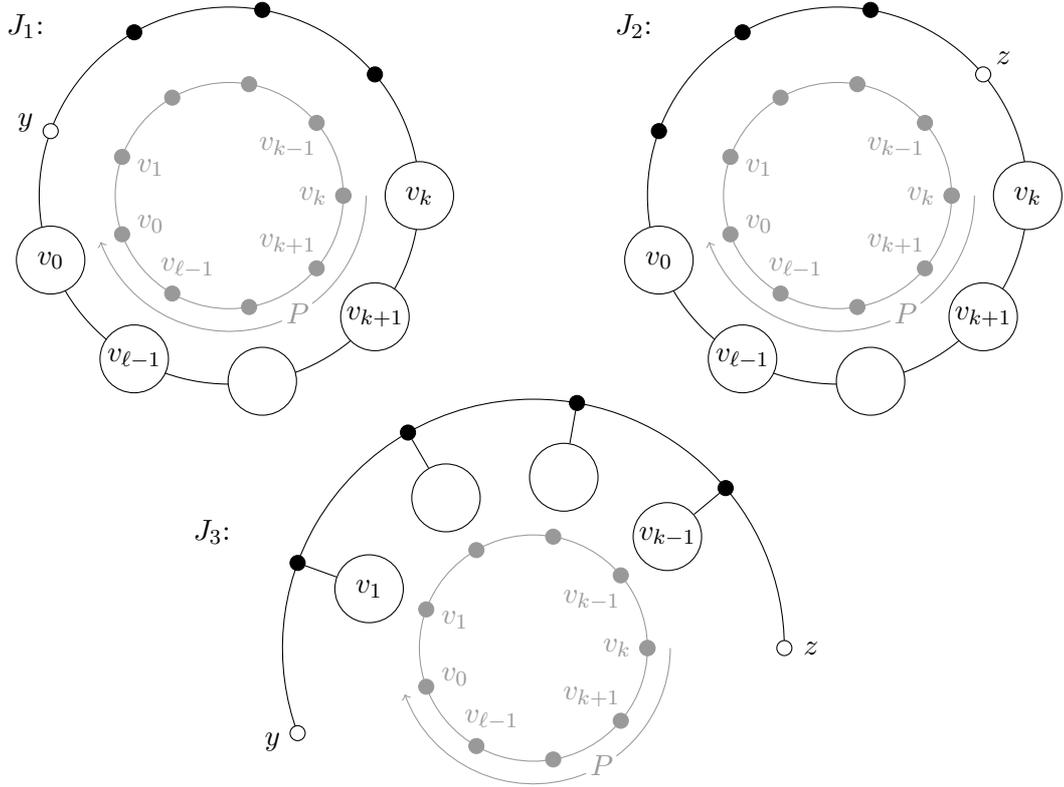

\begin{lemma}\label{lem:odd-cycle}
    Let $H$ be a square-free graph in which every vertex
    has odd degree.  If $H$~contains an odd cycle, then it has a
    hardness gadget.
\end{lemma}
\begin{proof}
    Let $\ell$ be the odd-girth of~$H$.
    If $H$~contains an edge in an odd number of $\ell$-cycles (which is
    guaranteed   for $\ell=3$, since $H$ is square-free), then $H$~has a hardness gadget by
    Lemma~\ref{lem:edge-in-odd}.  So, for the remainder of the proof,
    we may assume that the
    shortest odd cycle in~$H$ has length $\ell>4$ and that every edge is
    in a (not necessarily positive) even number of $\ell$-cycles.

    Let $P=v_k v_{k+1}\dots v_{\ell-1}v_0$ be a longest path that is
    in a positive, even number of $\ell$-cycles (see
    Figure~\ref{fig:gadget-from-cycle}; it turns out to be most
    convenient to label the vertices in this order; the path has
    length $\ell-k$).  Such a path certainly exists because any edge
    in an $\ell$-cycle is in a positive, even number of them. So, in
    particular, $P$~contains at least one edge.  Further $P$~has fewer
    than $\ell-1$ edges, because any path on $\ell-1$ edges is in at
    most one $\ell$-cycle, since $H$~has no parallel edges.  Let
    $C=v_0v_1\dots v_{\ell-1}v_0$ be an $\ell$-cycle containing~$P$.
    Let $\mathrm{rev}(P) = v_0v_{\ell-1} \dots v_k$ be the path~$P$
    with the vertices listed in the reverse order.

    Let $i=v_1$ and $s=v_{k-1}$. Let $J_1$~be the $\ell$-cycle gadget
    $J_{\ell,\mathrm{rev}(P),y}$, let $J_2$~be the $\ell$-cycle gadget
    $J_{\ell,P,z}$, and let
    $J_3$~be the caterpillar gadget $J_{v_0\dots v_k}$.

    We claim that $(i, s, (J_1, y), (J_2, z), (J_3, y, z))$ is a
    hardness gadget for~$H$.  Since $P$~was chosen to be a longest
    path in a positive, even number of $\ell$-cycles, any path $uP$
    in~$H$ must be in an odd number of $\ell$-cycles or in none at
    all.  Since $P$~itself is in an even number of $\ell$-cycles, the
    number of extensions~$uP$ in an odd number of cycles must be even.
    By Corollary~\ref{cor:Jcycle}, $|\Homs{(J_{\ell,P,z},z)}{(H,u)}|$ is
    the number of $\ell$-cycles in~$H$ that contain the path~$uP$.
    Therefore, $\Oz$~is precisely the set of
    vertices~$u$ such that $uP$ is in an odd number of $\ell$-cycles,
    so we have established that $|\Oz|$~is even.  Since $sP$~is an
    extension of~$P$, it is not in a positive, even number of
    $\ell$-cycles; it is in at least one $\ell$-cycle (namely, $C$) so
    it is in an odd number of them.  Therefore, $s\in\Oz$.  Similarly,
    $|\Oy|$~is even and $i\in\Oy$.

    It remains to verify that the conditions of
    Lemma~\ref{lem:J3-caterpillar} hold for~$J_3$, so that lemma gives
    us the remaining properties we need from
    Definition~\ref{defn:hardness-gadget}.  All vertices in~$H$ have
    odd degree by assumption, including in particular the interior
    vertices of~$P$.  We have already established that $i = v_1 \in
    \Oy$ and $s = v_{k-1} \in \Oz$.  Finally, $\Oy\subseteq
    \Gamma_H(v_0)$ because, in $G(J_1)$, $y$~is adjacent to a vertex
    that is pinned to~$v_0$.  Similarly, $\Oz\subseteq \Gamma_H(v_k)$.
\end{proof}

\subsection{Bipartite graphs}
\label{sec:gadget:bipartite}

The only remaining case is bipartite graphs~$H$ in which every vertex
has odd degree.  We show that, if $H$~has an ``even gadget'',
it has a hardness gadget.  And it turns out that 
every connected bipartite graph with more than one edge has an even gadget.

\begin{definition}
\label{defn:even-gadget}
An \emph{even gadget} for a bipartite graph~$H$ with at least one edge
is an edge $(a,b)$ of $H$ 
together with    
a connected
    bipartite graph~$G$ with a distinguished edge $(w,x)$ such that
    $|\Homs{(G,w,x)}{(H,a,b)}|$ is even.
\end{definition}

Note that, for bipartite $G$ and~$H$, 
with edges $(w,x)$ and $(a,b)$, respectively,
there is
always at least one homomorphism from $(G,w,x)$ to $(H,a,b)$, since
the whole of~$G$ can be mapped to the edge $(a,b)$.  So, although
Definition~\ref{defn:even-gadget} only requires
$|\Homs{(G,w,x)}{(H,a,b)}|$ to be even, the number of homomorphisms is
always non-zero.

Suppose that
 $H$ is any connected bipartite graph 
with more than one edge such
that, for some edge $(a,b)$ of $H$, 
$(H,a,b)$ is involution-free.
We will show that $H$
 has an even gadget.  If, furthermore, $H$~is
square-free, this even gadget gives a hardness gadget.  If $H$~is also
involution-free, the hardness gadget implies \parp{}-completeness
of \parhcol{}, by Theorem~\ref{thm:hardness-gadget}.

\begin{lemma}
\label{lem:always-even-gadget}
Suppose that $H$ is a connected bipartite graph with more than one edge such that,
for some edge $(a,b)$ of $H$, $(H,a,b)$ is involution-free. Then $H$ has an even gadget.
\end{lemma}
\begin{proof}
Let $H$ be a graph satisfying the conditions in the statement of the lemma.
Let $K_2$ be the graph consisting of the single edge $(a,b)$.
Clearly, $(K_2,a,b)$ is involution-free (since there are no non-trivial automorphisms of~$K_2$
that fix~$a$ and~$b$) and $H \not\isoto K_2$ since $H$ has
more than one edge, so $(H,a,b) \not\isoto (K_2,a,b)$.
By Corollary~\ref{cor:Lovasz} (taking $H'=K_2$ and $\ybar = \ybar' = (a,b)$),   
 there is a connected graph $(G,w,x)$ with distinguished vertices $w$ and $x$ such that
 $(w,x)$ is an edge 
 and  
\begin{equation}   |\Homs{(G,w,x)}{(H,a,b)}|
            \not\equiv |\Homs{(G,w,x)}{(K_2,a,b)}| \pmod2\,.\label{eq:August}\end{equation}
            
$G$ must be bipartite --- otherwise
$$   |\Homs{(G,w,x)}{(H,a,b)}|
             =  |\Homs{(G,w,x)}{(K_2,a,b)}| =0\,,$$
             contradicting~\eqref{eq:August}.
Thus,
 $ |\Homs{(G,w,x)}{(K_2,a,b)}|=1$,       
so the edge $(a,b)$ of~$H$ together with $(G,w,x)$ 
 is an even gadget. \end{proof}

\begin{lemma}
\label{lem:even-gadget}
Suppose that $H$ is a connected, bipartite, 
square-free
graph with more than one edge such that,
for some edge $(a,b)$ of $H$, $(H,a,b)$ is involution-free. 
Suppose that every vertex of~$H$ has odd degree.
Then $H$ has a hardness gadget.
\end{lemma}
\begin{proof}
    By Lemma~\ref{lem:always-even-gadget}, $H$~has an even gadget.
 Choose
an even
gadget consisting of an edge $(i,s)$ of~$H$
 and a connected bipartite graph $G$ with distinguished edge $(w,x)$
 so that  $N = |\Homs{(G,w,x)}{(H,i,s)}|$ is even.
 Choose the even gadget so that the number of vertices of~$G$
 is as small as possible.
    There is a homomorphism from~$G$ to the edge $(i,s)$ so $N>0$.
    $N$~is even, so $G$~cannot be a single edge.

    First, we show that $\deg_G(w)\geq 2$.  Suppose, towards a
    contradiction, that $\deg_G(w)=1$, i.e., that $x$~is the only
    neighbour of~$w$ in~$G$.  If this is the case, then $x$~must have
    some neighbour~$w'\neq w$, since $G$~is not a single edge.  We
    have
    \begin{align*}
        0&\equiv |\Homs{(G,w,x)}{(H,i,s)}| \pmod2 \\
         &\equiv |\Homs{(G-w,x)}{(H,s)}| \pmod2   \\
         &= \!\!\!\sum_{c\in\Gamma_{\!H}(s)}\!\!\! |\Homs{(G-w,x,w')}{(H,s,c)}|\,.
    \end{align*}
    Since every vertex in~$H$ has odd degree, the sum has an odd
    number of terms.  Since the total is even, there must be some~$c$
    such that $|\Homs{(G-w,x,w')}{(H,s,c)}|$ is even, contradicting
    the choice of~$G$.  By the same argument, $\deg_G(x)\geq 2$, also.

    For any vertex~$v\in V(G)$, let
    \begin{equation*}
        C(v) = \{c\in V(H) \mid |\Homs{(G,w,x,v)}{(H,i,s,c)}|
                                    \text{ is odd}\}\,.
    \end{equation*}
    Note that, for any $v\in V(G)$, $|C(v)|$~is even since, otherwise,
    $N$~would be odd.

    We now show that $C(y)\neq\emptyset$ for every $y\in
    \Gamma_{\!G}(x) \setminus \{w\}$.  If $C(y)=\emptyset$, then, in
    particular, $i\notin C(y)$, so $|\Homs{(G,w,x,y)}{(H,i,s,i)}|$ is
    even.  But then $|\Homs{(G'\!,w,x)}{(H,i,s)}|$ is even, where
    $G'$~is the graph made from $G$ by identifying the (distinct)
    vertices $w$ and~$y$ and calling the resulting vertex~$w$.  This
    contradicts minimality in the choice of~$G$.
    Similarly, $C(z)\neq\emptyset$ for
    every $z\in \Gamma_{\!G}(w) \setminus \{x\}$.  Choose vertices
    $y\in \Gamma_{\!G}(x) \setminus \{w\}$ and $z\in \Gamma_{\!G}(w)
    \setminus \{x\}$.

    Finally, let $J$ be the partially $H$-labelled graph $(G,
    \{w\mapsto i, x\mapsto s\})$ and let $G(J_3)$~be the single edge
    $(y,z)$ and $\tau(J_3)=\emptyset$.
    We show that
    $(i, s, (J,y), (J,z), (J_3,y,z))$ is a hardness gadget for~$H$.  $\Oy =
    C(y)$ is even and $i\in C(y)$; likewise, $\Oz=C(z)$ is even and
    $s\in C(z)$.

    By the choice of~$J$, $\Oy\subseteq \Gamma_{\!H}(s)$
    and $\Oz\subseteq \Gamma_{\!H}(i)$.  For any $o\in\Oy-i$ and
    $x\in\Oz-s$, $H$~contains edges $(o,s)$, $(s,i)$ and~$(i,x)$ so it
    does not contain the edge $(o,x)$ as it is square-free.
    Therefore, $|\Sos| = |\Sis| = |\Six| = 1$ and $|\Sox| = 0$ and we
    have established all the conditions of
    Definition~\ref{defn:hardness-gadget}.
\end{proof}

\section{Main theorem}
\label{sec:mainthm}

We have shown that all connected, square-free, involution-free graphs
(and some  disconnected graphs, too) have hardness gadgets and
that \parhcol{} is \parp{}-complete for any involution-free
graph that has a hardness gadget.  To deal with
graphs that have involutions, we use  
reduction by involutions. 
As we noted in the introduction, 
Faben and Jerrum showed that every graph~$H$ has
a unique (up to isomorphism) involution-free reduction~$H^*$.
 They also proved~\cite[Theorem 3.4]{FJ13} 
that for any graph~$G$,    
$|\Homs{G}{H}| \equiv |\Homs{G}{H^*}|
    \bmod2$. Hence,  \parhcol{} has the same complexity as $\parhcol[H^*]$.  
  
If $H$~is a tree (as it was for Faben and Jerrum), then 
its involution-free reduction $H^*$~is connected.  
However, for general
graphs, the fact that $H$ is connected does not imply that $H^*$~is connected.\footnote{For
example, consider non-isomorphic, disjoint, connected, involution-free
graphs $H_1$
and~$H_2$ and let $H$~be a graph made by adding two disjoint paths of
the same length from some vertex $x_1\in H_1$ to some vertex $x_2\in
H_2$.  The only involution of this graph exchanges the interior
vertices of the two paths, so $H^* = H_1\cup H_2$, which is
disconnected.} The final result that we need from Faben and Jerrum is
\cite[Theorem~6.1]{FJ13}, which allows us to deal with disconnected
graphs:

\begin{lemma}\label{lem:disconnected}
    Let $H$~be an involution-free graph.  If $H$~has a component~$H'$
    for which \parhcol[H'] is \parp{}-complete, then \parhcol{}
    is \parp{}-complete.
\end{lemma}

We can now prove our main result.

{
\renewcommand{\thetheorem}{\ref{thm:main}}
\begin{theorem}
\statethmmain{}
\end{theorem}
}
\begin{proof}

    As we noted above,   \parhcol{}
    has the same complexity as $\parhcol[H^*]$.
   If $H^*$~has at most one 
vertex, then
\parhcol[H^*] is in~\Ptime{}: $|\Homs{G}{H^*}|=1$
if $G$ has no edges and
$\Homs{G}{H^*}=\emptyset$ 
if $G$~has an edge.
    Otherwise, let $H^{**}$~be any component of~$H^*$
    with more than one vertex.  Such a component must exist since,
    otherwise, $H^*$~would be a graph with at least two vertices and
    no edges, and any such graph has an involution.

    If $H^{**}$~has two or more vertices of even degree, then it has a
    hardness gadget by Lemma~\ref{lem:two-even}.  If $H^{**}$~has exactly
    one vertex of even degree, it has a hardness gadget by
    Lemma~\ref{lem:even-deg}.
    If the previous cases do not apply, then every vertex of~$H^{**}$
    must have odd degree.  By Lemma~\ref{lem:invo-free-cycle},
    $H^{**}$~contains a cycle.  If it contains an odd cycle, it
    has a hardness gadget by Lemma~\ref{lem:odd-cycle}.  Otherwise,
    $H^{**}$~is bipartite.
    By construction, $H^{**}$ is connected and square-free.
 Since $H^{**}$ contains a cycle, it has more than one edge.
 Since it is involution-free, it certainly contains an edge $(a,b)$
 so that $(H^{**},a,b)$ is involution-free.   Every vertex of~$H^{**}$ has
 odd degree, so
    it has a hardness gadget by Lemma~\ref{lem:even-gadget}.

    We have established that either $H^*$ has at most one vertex, in
    which case \parhcol[H^*] and \parhcol{} are in~\Ptime{}, or that
    some component $H^{**}$ of $H^*$~has a hardness gadget.  In the
    latter case, \parhcol[H^{**}] is \parp{}-complete by
    Theorem~\ref{thm:hardness-gadget}.
    \parhcol[H^*] is \parp{}-complete by Lemma~\ref{lem:disconnected},
    so \parhcol{} is \parp{}-complete.
\end{proof}

 \section{Acknowledgements}
 We thank the referees for very useful comments.

\bibliographystyle{plain}
\bibliography{\jobname}

\end{document}